\documentclass[10pt]{article}

\usepackage{amsmath}
\usepackage{amsfonts}
\usepackage{amsthm}
\usepackage[english]{babel}

\newtheorem{theorem}{Theorem}
\newtheorem{lemma}{Lemma}
\newtheorem{definition}{Definition}
\newtheorem{corollary}{Corollary}

\newtheorem{remark}{Remark}

\newcommand{\nd}{{\rm and }}

\newcommand{\mR}{\mathbb{R}}

\newcommand{\mN}{\mathbb{N}}
\newcommand{\mE}{\mathbb{E}}
\newcommand{\mZ}{\mathbb{Z}}
\newcommand{\mS}{\mathbb{S}}

\newcommand{\mB}{\mathbb{B}}

\newcommand{\cM}{\mathcal{M}}
\newcommand{\cH}{\mathcal{H}}

\newcommand{\cP}{\mathcal{P}}

\newcommand{\cS}{\mathcal{S}}

\newcommand{\cL}{\mathcal{L}}

\newcommand{\ux}{\underline{x}}

\newcommand{\uxi}{\underline{\xi}}

\newcommand{\px}{\partial_x}

\newcommand{\upx}{\partial_{\underline{x}}}

\begin{document}
\title{$q$-deformed harmonic and Clifford analysis and the $q$-Hermite and Laguerre polynomials}

\author{K.\ Coulembier\thanks{Corresponding author} \thanks{Ph.D. Fellow of the Research Foundation - Flanders (FWO), E-mail: {\tt Coulembier@cage.ugent.be}},  F.\ Sommen\thanks{E-mail: {\tt fs@cage.ugent.be}}}

\date{\small{Clifford Research Group -- Department of Mathematical Analysis}\\
\small{Faculty of Engineering -- Ghent University\\ Krijgslaan 281, 9000 Gent,
Belgium}}

\maketitle

\begin{abstract}
We define a $q$-deformation of the Dirac operator, inspired by the one dimensional $q$-derivative. This implies a $q$-deformation of the partial derivatives. By taking the square of this Dirac operator we find a $q$-deformation of the Laplace operator. This allows to construct $q$-deformed Schr\"odinger equations in higher dimensions. The equivalence of these Schr\"odinger equations with those defined on $q$-Euclidean space in quantum variables is shown. We also define the $m$-dimensional $q$-Clifford-Hermite polynomials and show their connection with the $q$-Laguerre polynomials. These polynomials are orthogonal with respect to an $m$-dimensional $q$-integration, which is related to integration on $q$-Euclidean space. The $q$-Laguerre polynomials are the eigenvectors of an $su_q(1|1)$-representation.
\end{abstract}


\newpage
\section{Introduction}

Quantum algebras are $q$-deformed versions of universal enveloping algebras of Lie algebras, the latter are recovered as the deformation parameter $q$ goes to unity. The study of quantum algebras leads to the use of mathematical tools of $q$-analysis, see \cite{MR0708496, MR1052153}. In \cite{JACKSON} Jackson originally introduced the $q$-analogues of differentiation, integration and special functions in the context of $q$-hypergeometric series (also known as basic hypergeometric series). In particular, there are connections between representations of quantum algebras and $q$-special functions (\cite{MR1121848}) and $q$-calculus (\cite{MR128259}). Also in the framework of $q$-harmonic analysis of this paper, we will obtain an $su_q(1|1)$-representation for which the $q$-Laguerre polynomials (see \cite{MR0708496}) are eigenvectors.

In \cite{MR1288667} $q$-analysis is used to solve the $SO_q(m)$-invariant Schr\"odinger equation in quantum Euclidean space, see \cite{MR1239953}. The results in \cite{MR2316312} about $q$-difference equations can be used to solve more general $SO_q(m)$-invariant Schr\"odinger equations in quantum Euclidean space. Also the objects of $q$-harmonic analysis developed in this article can be used to study quantum Euclidean space.

The interest for quantum algebras in physics was partly triggered by the introduction of the $q$-deformed harmonic oscillator (see \cite{Bonatsos} for an overview). The first approaches however, lacked any dynamical content behind the hamiltonian. In \cite{MR128259} an overview is given of different realizations of the $q$-Heisenberg algebra, using the $q$-derivative, leading to the $q$-harmonic oscillator. In \cite{MR2455812} a procedure for general $q$-deformed quantum mechanics was constructed using the $q$-derivative. Until now the higher dimensional $q$-deformed isotropic oscillator is only defined in quantum Euclidean space or in undeformed space by an unnatural separation of the radial part.

The $q$-harmonic oscillators lead to the $q$-Hermite polynomials see e.g. \cite{MR1448659}. In \cite{MR0708496, MR0599807} it was shown that the $q$-Hermite polynomials are orthogonal with respect to $q$-integration and have annihilation and creation operators using the $q$-derivative. Because all the different types of the $q$-Hermite polynomials satisfy many properties that are analogues of properties of the Hermite polynomials, see e.g. \cite{MR1448659, MR0599807, MR2392899}, they are interesting objects of study themselves.

In this study we define a theory of $q$-deformed derivatives in higher dimensions and a $q$-deformed Laplace operator acting on functions with commuting variables. Since we use Clifford analysis (\cite{MR697564, MR1169463}) for this, we define a $q$-deformed Dirac operator with its square a $q$-Laplace operator. As the undeformed $SO(m)$-invariant harmonic operators generate the Lie algebra $\mathfrak{sl}_2(\mR)$ we now find a $q$-deformation of this algebra. This leads to a Howe dual pair (\cite{MR0986027}) with a quantum algebra, $\left(SO(m),\mathfrak{sl}_2(\mR)_q\right)$. 

Because the $q$-Laplace operator is scalar, it can be expressed without Clifford algebras. However, the Clifford-approach to this $q$-Laplace operator is more natural. The resulting $q$-Laplace operator can be used to put an existing $q$-deformation of the isotropic Schr\"odinger equation in undeformed space (\cite{MR2316312, MR1447846, Papp}) in a complete setting. This equation has its origin in quantum Euclidean space, see \cite{MR1288667, MR1239953, MR1097174}. Using the $q$-Laplace operator, the angular and radial part are reunited in a complete Schr\"odinger equation in undeformed space. This quantum system has the same energy spectrum as the Schr\"odinger equation in quantum Euclidean space. 

Using the $q$-deformed Dirac operator we can define a $q$-deformation of the Clifford-Hermite polynomials. These are higher dimensional generalizations of the one dimensional Hermite polynomials, see \cite{MR926831}. Similar to the undeformed case there is a connection between the $q$-Clifford Hermite polynomials and the one dimensional $q$-Laguerre polynomials which were introduced in \cite{MR0486697, MR0618759}. Once again, this justifies the choice of our $q$-Dirac operator. Using this construction we obtain realizations of $\mathfrak{su}(1|1)_q$ acting on $\mR[t]$ for which the $q$-Laguerre polynomials are the eigenvectors. This is a concrete generalization of the occurrence of the Laguerre polynomials as formal eigenvectors in the $l^2(\mZ_+)$ representation space for $\mathfrak{su}(1|1)$, see e.g. \cite{MR1447890}.

The paper is organized as follows. First we repeat some facts about quantum numbers and derivatives and give a short introduction to Clifford analysis. From a list of axioms we derive a unique $q$-Dirac operator, which leads to a $q$-Laplace operator. We show an important connection with the $SO(m)_q$-invariant Laplace operator in $q$-Euclidean space. Then we construct an integration which leads to $q$-Cauchy formulas. Finally we define the $q$-Clifford-Hermite polynomials and prove their most important properties.  

\section{Preliminaries}
\label{preliminaries}

We give a short introduction to quantum numbers ($q$-numbers), $q$-derivatives and $q$-integration, see \cite{MR0708496, MR1052153, JACKSON, Bonatsos}. For $u$ a number or operator, and $q$ the deformation parameter, we define (where it exists) the $q$-deformation of $u$ by
\begin{eqnarray*}
[u]_q=\frac{q^u-1}{q-1}.
\end{eqnarray*}

It is clear that $\lim_{q\to 1}[u]_q=u$. In this paper we assume $q\in\mR^+$. The $q$-derivative of a function $f(t)$ is defined by
\begin{eqnarray}
\label{qder}
\partial^q_t (f(t))&=&\frac{f(qt)-f(t)}{(q-1)t}.
\end{eqnarray}

For this to exist, the function has to be defined in $t$ and $qt$ and has to be differentiable in the origin. From the definition we find
\begin{equation*}
\partial^q_t (t^k)=\frac{q^k-1}{q-1}t^{k-1}=[k]_qt^{k-1}
\end{equation*}

and the Leibniz rule
\begin{eqnarray}
\label{qderiv}
\partial^q_t t=qt \partial^q_t+1.
\end{eqnarray}
This is a special case of the following two Leibniz rules
\begin{eqnarray}
\label{Leibniz}
\partial_t^q(f_1(t)f_2(t))&=&\partial_t^q(f_1(t))f_2(t)+f_1(qt)\partial_t^q(f_2(t))\\
\label{Leibniz2}
&=&\partial_t^q(f_1(t))f_2(qt)+f_1(t)\partial_t^q(f_2(t)).
\end{eqnarray}

For $q<1$, the $q$-integration on an interval $[0,a]$ with $a\in \mR$ is given by
\begin{eqnarray}
\label{1Dqint}
\int_{0}^af(t)\,d_qt&=&(1-q)a\sum_{k=0}^\infty f(aq^k)q^k.
\end{eqnarray}

More general intervals are defined by $\int_a^b=\int_0^b-\int_0^a$ and satisfy the important property
\begin{eqnarray}
\label{partiele}
\int_{a}^b\left(\partial_t^q f\right)(t)\,d_qt&=&f(b)-f(a).
\end{eqnarray}

The $q$-factorial of an integer $k$ is given by $[k]_q!=[k]_q[k-1]_q\cdots [1]_q.$ This leads to the introduction of the $q$-exponential
\begin{eqnarray}
\label{exp1}
E_q(t)&=&\sum_{j=0}^\infty\frac{t^j}{[j]_q!}.
\end{eqnarray}
In order to find its inverse we define a second $q$-exponential by
\begin{eqnarray}
\label{exp2}
e_q(t)=E_{q^{-1}}(t)&=&\sum_{j=0}^\infty q^{\frac{1}{2}j(j-1)}\frac{t^j}{[j]_q!}.
\end{eqnarray}
Now $E_q(t)e_q(-t)=1$, see \cite{MR1052153, MR1121848, Bonatsos}. It is easily calculated that
\begin{eqnarray}
\label{aflexp}
\partial_t^qE_q(t)=E_q(t)&, &\partial_t^qe_q(t)=e_q(qt).
\end{eqnarray}
The series $E_q(t)$ converges absolutely and uniformly everywhere if $q>1$ and for $|t|<\frac{1}{1-q}$ if $q<1$, see \cite{MR0708496}. The $q$-binomial coefficients are defined by
\begin{eqnarray*}
\binom{n}{k}_q&=&\frac{[n]_q!}{[n-k]_q![k]_q!}.
\end{eqnarray*}

We can also define the $q$-Gamma function for $q<1$
\begin{eqnarray*}
\Gamma_q(t)&=&\frac{\prod_{k=1}^\infty(1-q^k)}{\prod_{k=0}^\infty(1-q^{t+k})}(1-q)^{1-t},
\end{eqnarray*}

with the property that $\Gamma_q(t+1)=[t]_q\Gamma_q(t)$, see \cite{MR1052153}. The function $\Gamma_q$ admits the following integral representation
\begin{eqnarray}
\label{intGamma}
\Gamma_q(z)&=&\int_0^{\frac{1}{1-q}}d_qt\,t^{z-1}e_{q}(-qt),
\end{eqnarray}
see \cite{MR0708496, MR1121848}. Sometimes we will encounter expressions which we will write as $q^2$-deformations, for example $[2u]_q=(q+1)[u]_{q^2}$, therefore we fix the notation $Q=q^2$.

Now we briefly recall the basic notions of Clifford analysis. For more details we refer the reader to \cite{MR697564, MR1169463}. Denote by $\mR_{0,m}$ the Clifford algebra generated by an orthonormal basis $(e_1,\cdots, e_m)$ for $\mR^m$ with multiplication rules
\begin{eqnarray}
\label{Clifford}
e_ie_j+e_je_i&=&-2\delta_{ij}
\end{eqnarray}

for $1\le i,j\le m$. The algebra generated by these Clifford numbers and the $m$ commuting variables $x_j$, which commute with $e_i$, $1\le i\le m$, is the algebra of Clifford valued polynomials $\cP=\mR[x_1,\cdots, x_m]\otimes\mR_{0,m}$. The vector variable is identified with the first order Clifford polynomial of the form $\ux=\sum_{j=1}^me_j{x_j}$. Using (\ref{Clifford}) we find that the square of this vector variable is scalar valued, $\ux^2=-\sum_{j=1}^mx_j^2=-r^2$. The corresponding vector derivative in the vector variable $\ux $ is the Dirac operator,
\begin{equation*}
\label{diracoperator}
\upx=-\sum_{j=1}^me_j\partial_{x_j}.
\end{equation*}

The square of the Dirac operator is again scalar, $\upx^2=-\Delta$, with $\Delta$ the Laplace operator. Using the Clifford multiplication rules (\ref{Clifford}) we can calculate
\begin{eqnarray}
\label{anticomxpx}
\{\ux,\upx\}=\upx \ux+\ux\upx&=&2\mE+m,
\end{eqnarray}
with $\mE=\sum_{j=1}^mx_j\partial_{x_j}$ the Euler operator. In particular we find $\upx (\ux)=m$ and
\begin{eqnarray}
\label{comxxpx}
\upx \ux^2&=&\ux^2\upx +2 \ux.
\end{eqnarray}

We will use the notation $f(\ux)=f(x_1,\cdots,x_m)$. Clifford analysis deals with the function theory of solutions of $\upx f(\ux)=0$, called monogenic functions, in particular monogenic polynomials of degree $k$. 

\begin{definition}
An element $F \in \cP$ is a spherical monogenic of degree $k$ if it satisfies
\begin{eqnarray*}
\upx F =0 &\mbox{and}& \mE F = kF.
\end{eqnarray*}

The space of all spherical monogenics of degree $k$ is denoted by $\cM_k$.
\end{definition}

In the same way we can define the space of spherical harmonics of degree $k$, $\cH_k$, as the null solutions of the Laplace operator, clearly $\cM_k\subset\cH_k$. We have the following well-known decomposition of the space of polynomials. 

\begin{lemma}[Fischer decomposition I]
The vector space $\cP_k$ decomposes as
\begin{equation*}
\cP_k = \bigoplus_{i=0}^{\lfloor k/2 \rfloor} \ux^{2i} \cH_{k-2i}.
\end{equation*}
This decomposition is unique, hence $\sum_{i}\ux^{2i}H_{k-2i}=0$, with ($H_{j}\in\cH_{j}$) implies $H_{k-2i}=0$ for every $i$.
\label{fischerdecomp}
\end{lemma}

Since $\Delta$ is scalar, we can replace $\cP$ with $\mR[x_1,\cdots,x_m]$ in the previous lemma. The decomposition can be refined to

\begin{lemma}[Fischer decomposition II]
The vector space $\cP_k$ decomposes as
\begin{equation*}
\cP_k = \bigoplus_{i=0}^{k} \ux^{i} \cM_{k-i}.
\end{equation*}
This decomposition is unique, hence $\sum_{j}\ux^{k-j}M_{j}=0$ (with $M_{j}\in\cM_{j}$) implies $M_{j}=0$ for every $j$.
\label{cliffordfischerdecomp}
\end{lemma}

Commutation rules \eqref{anticomxpx} and \eqref{comxxpx} yield
\begin{eqnarray}
\label{1dirac1}
\upx \ux^{2l}M_k&=&2l\ux^{2l-1}M_k\\
\label{1dirac2}
\upx \ux^{2l+1}M_k&=&(2l+2k+m)\ux^{2l}M_k
\end{eqnarray}

These equations together with lemma \ref{cliffordfischerdecomp} imply that every (scalar) $H_k$ can be decomposed as
\begin{eqnarray}
\label{decompHk}
H_k&=&M_k+\ux M_{k-1}.
\end{eqnarray}

The operators $\upx$ and $\ux$ generate a finite-dimensional Lie super-algebra isomorphic to $\mathfrak{osp}(1|2)$. The even subalgebra is generated by $\upx^2$, $\ux^2$ and $\mE+m/2$ and is isomorphic to the Lie algebra $\mathfrak{sl}_2(\mR)$, see \cite{MR0986027, DBE1}. The commutation relations of the Lie super-algebra are given by
\[
\begin{array}{lllllll}
 \left[\upx^2/2, \ux^2/2\right] &=&\mE+m/2 &\;&\{\ux,\ux\} &=&2\ux^2 \\
\left[\upx^2/2, \mE+m/2\right] &=&2\upx^2/2 &\;&\{\upx, \upx \} &=&2\upx^2 \\
 \left[\ux^2/2, \mE + m/2 \right] &=&-2\ux^2/2&\;&\{\upx, \ux\} &=&2\mE+m\\
\end{array}
\]
and
\[
\begin{array}{lllllll}
\left[\ux, \ux^2\right] &=&0 &\;& \left[\upx, \ux^2\right] &=&2\ux\\
\left[\ux, \upx^2 \right] &=&-2\upx &\;& \left[\upx, \upx^2\right] &=&0\\
\left[\ux, \mE + m/2\right] &=&-\ux &\;& \left[\upx, \mE + m/2 \right] &=&\upx.\\
\end{array}
\]
An important feature in harmonic and Clifford analysis is the occurrence of Howe dual pairs, see \cite{MR0986027}. The generators of the Lie algebra $\mathfrak{sl}_2(\mR)$ are $SO(m)$-invariant. These operators acting on the module $\oplus_{j}r^{2j}\cH_k$ give an infinite-dimensional irreducible representation of $\mathfrak{sl}_2(\mR)$. The blocks $r^{2j}\cH_k'$, with $\cH_k'$ the scalar spherical harmonics, are the irreducible pieces of $\mR[x_1,\cdots,x_m]$ under the action of $SO(m)$. This can be refined to the Howe dual pair $\left(\mbox{Spin}(m),\mathfrak{osp}(1|2)\right)$ (see \cite{MR697564, MR1169463}), with Spin$(m)$ the universal cover of $SO(m)$.

These Howe dual pairs return in different generalizations of harmonic and Clifford analysis. In Dunkl harmonic analysis (see \cite{OSS}) we have the pair $\left(G,\mathfrak{sl}_2(\mR)\right)$ with $G$ a Coxeter group. In super  harmonic analysis (see \cite{DBE1}) we find the Howe dual pair $\left(SO(m)\times Sp(2n),\mathfrak{sl}_2(\mR)\right)$. The Howe dual pair for hermitian Clifford analysis can be found in \cite{herm}. By defining $q$-deformed Clifford analysis we will obtain Howe dual pairs with the quantum algebras $\mathfrak{sl}(\mR)_q$ and $\mathfrak{osp}(1|2)_q$.

The Euler operator $\mE=r\partial_r$ represents the radial part in $\ux\upx$, the angular part is given by the Gamma operator $\Gamma$,
\begin{eqnarray}
\label{pxEGamma}
\ux\upx&=&\mE+\Gamma.
\end{eqnarray}
By using $\ux=r\uxi$ with $\uxi^2=-1$, this equation can also be written as $\upx=-\underline{\xi}(\partial_r+\frac{1}{r}\Gamma)$. While the Euler operator is scalar, the Gamma operator is a bivectorial operator, $\Gamma=-\sum_{i<j}e_ie_j(x_i\partial_{x_j}-x_j\partial_{x_i})$. Using $\mE\ux=\ux\mE+\ux$ and (\ref{anticomxpx}) we obtain the commutation relations for the Gamma operator,
\begin{eqnarray}
\label{1Gamma1}
\Gamma \ux&=&\ux(m-1-\Gamma)\\
\label{1Gamma2}
\Gamma \ux^{2}&=&\ux^{2}\Gamma.
\end{eqnarray}

We will also need the main anti-involution on the Clifford algebra $\mR_{0,m}$, defined by
\begin{eqnarray*}
\overline{e_i} &=&-e_i\\
\overline{a b} &=& \overline{b} \overline{a}, \quad \mbox{for all } a,b \in \mR_{0,m}.
\end{eqnarray*}

For Clifford valued functions on the unit sphere there is an inner product
\begin{eqnarray*}
\langle f|g\rangle=\int_{\mS^{m-1}}d\xi\,\,\,\left[\,\overline{f}\,g\,\right]_0,
\end{eqnarray*}
 
with $\overline{\cdot}$ the main anti-involution and $[\cdot]_0:\mR_{0,m}\to\mR$, the projection onto the scalar part. For two spherical harmonics of degree $k\not=l$,
\begin{eqnarray}
\label{HH}
\int_{\mS^{m-1}}d\xi\,\,\,H_kH_l&=&0
\end{eqnarray}
holds. In particular, we will consider a fixed orthonormal basis of spherical monogenics $M_k^{(p)}$,
\begin{eqnarray}
\label{orthbasis}
\int_{\mS^{m-1}}d\xi\,\,\,\left[\overline{M_k^{(p)}}M_l^{(r)}\right]_0&=&\delta_{kl}\delta_{pr}.
\end{eqnarray}

\section{Definition of the operators}

\subsection{The $q$-Dirac operator}

Our aim is to obtain a $q$-deformed version of the vector derivative, or Dirac operator $\upx$ which we will denote by $\upx^q$. First we derive 4 axioms such an operator should satisfy. Inspired by $\partial_t^q(t)=1=[1]_q$ and formula (\ref{anticomxpx}) we impose $\upx^q(\ux)=[m]_q$. We also need a good $q$-deformed Leibniz rule based on (\ref{qderiv}). We deform commutation relation (\ref{comxxpx}) in stead of \eqref{anticomxpx} because $\ux^2$ is scalar. Therefore we can elegantly extend $\partial_t^qt^2=q^2t^2\partial_t^q+(q+1)t$ to
\begin{eqnarray*}
\upx^q\ux^2=q^2\ux^2\upx^q+(q+1)\ux.
\end{eqnarray*}

To obtain a $q$-deformation of the Laplace operator, $(\upx^q)^2$ has to be a scalar operator. For the last axiom we use the Fischer decomposition in lemma \ref{cliffordfischerdecomp} to find that a basis for the polynomials of degree one is given by
\begin{equation*}
\ux,\, x_1e_2+x_2e_1,\,\cdots\,,\,x_1e_m+x_me_1.
\end{equation*}
It can be shown  that all monogenic functions are Taylor series in the $x_je_1+x_1e_j$, $j\not= 1$ (see \cite{MR697564, MR1169463}). This is a generalization of the fact that holomorphic functions (null solutions of the Cauchy-Riemann operator $\partial_{\overline{z}})$ are Taylor series in $z$ and not in $\overline{z}$. Because $\upx^q$ should be a $q$-deformation of the derivative with respect to $\ux$ we do not want it to mix up with the derivation with respect to $x_1e_2+x_2e_1$. This means $\upx^q$ should satisfy
\begin{eqnarray*}
\upx^q f&=&0
\end{eqnarray*}

when $\upx f=0$. Summarizing, $\upx^q$ should satisfy the following $4$ axioms,
\begin{eqnarray*}
(A1)\;\; & &\;\;\;\; \upx^q(\ux)=[m]_q\\
(A2)\;\; & &\;\;\;\; \upx^q\ux^2=q^2\ux^2\upx^q+(q+1)\ux\\
(A3)\;\; & &\;\;\;\; (\upx^q)^2\,\mbox{is scalar}\\
(A4)\;\; & &\;\;\;\; \upx^qM_k=0.
\end{eqnarray*}

We will show that these axioms uniquely define the $q$-Dirac operator on $\cP$.
\begin{lemma}
\label{scalar}
A linear operator on $\cP$ satisfying $(A2)$ and $(A3)$ also satisfies the property that
\begin{eqnarray*}
\upx^q\ux+q^2\ux\upx^q
\end{eqnarray*}

is a scalar operator.
\end{lemma}
\begin{proof}
We calculate $(\upx^q)^2\ux^2$ using $(A2)$,
\begin{eqnarray*}
(\upx^q)^2\ux^2&=&q^2\upx^q\ux^2\upx^q+(q+1)\upx^q\ux\\
&=&q^4\ux^2(\upx^q)^2+q^2(q+1)\ux\upx^q+(q+1)\upx^q\ux.
\end{eqnarray*}

Rearranging terms yields
\begin{eqnarray}
\label{EE}
(\upx^q)^2\ux^2-q^4\ux^2(\upx^q)^2&=&(q+1)(\upx^q\ux+q^2\ux\upx^q).
\end{eqnarray}
Because $\ux^2$ and $(\upx^q)^2$ are scalar we obtain the lemma.
\end{proof}

\begin{lemma}
\label{pxxM}
For a linear operator on $\cP$ satisfying $(A1)-(A4)$, the following relation holds,
\begin{eqnarray*}
\upx^q\ux M_k&=&[m+2k]_qM_k.
\end{eqnarray*}

\end{lemma}
\begin{proof}
We know from $(A1)$ that this holds for $k=0$.  Now we assume that $\upx^q\ux M_{k}=[m+2k]_qM_k$ holds and prove that it also holds for $k+1$. Using $(A4)$ yields
\begin{eqnarray}
\label{tussthm}
\upx^q\ux M_{k+1}&=&(\upx^q\ux+q^2\ux\upx^q) M_{k+1}.
\end{eqnarray}
We use $(A2)$, $(A4)$ and the induction step to calculate
\begin{eqnarray*}
(\upx^q\ux+q^2\ux\upx^q) \ux M_k&=&(q+1+q^2[m+2k]_q)\ux M_k\\
&=&[m+2k+2]_q\ux M_k.
\end{eqnarray*}

Let $H_{k+1}$ be an arbitrary scalar spherical harmonic of degree $k+1$, this means $\ux\upx H_{k+1}\in\ux\cM_k$ and we can substitute $\ux\upx H_{k+1}$ for $\ux M_k$ in the equation above. Equation (\ref{pxEGamma}) implies
\begin{eqnarray*}
\ux\upx H_{k+1}=(k+1)H_{k+1}+\Gamma H_{k+1},
\end{eqnarray*}

so the scalar part of $\ux\upx H_{k+1}$ is proportional to $H_{k+1}$. Since $(\upx^q\ux+q^2\ux\upx^q) $ is a scalar operator (lemma \ref{scalar}) the equation above holds separately for both $H_{k+1}$ and $\Gamma H_{k+1}$. So for every scalar $H_{k+1}$
\begin{eqnarray*}
(\upx^q\ux+q^2\ux\upx^q)H_{k+1}&=&[m+2k+2]_qH_{k+1}.
\end{eqnarray*}

This equation can be multiplied with elements of the Clifford algebra on the right hand side, so it also holds for $M_{k+1}$. Combining this with equation (\ref{tussthm}) yields
\begin{eqnarray*}
\upx^q\ux M_{k+1}&=&[m+2k+2]_qM_{k+1},
\end{eqnarray*}

so the lemma is proven by induction.
\end{proof}

\begin{theorem}
\label{pxFischer}
There is at most one linear operator on $\cP$ satisfying $(A1)-(A4)$. The action on the Fischer decomposition (lemma \ref{cliffordfischerdecomp}) is given by
\begin{eqnarray*}
\upx^q\ux^{2l} M_k&=&[2l]_q\ux^{2l-1}M_k\\
\upx^q\ux^{2l+1}M_k&=&[2l+2k+m]_q\ux^{2l}M_k.
\end{eqnarray*}
\end{theorem}
\begin{proof}
Iterating $(A2)$ yields
\begin{eqnarray*}
\upx^q\ux^{2l}=[2l]_q\ux^{2l-1}+q^{2l}\ux^{2l}\upx^q.
\end{eqnarray*}

Together with lemma \ref{pxxM}, this proves the theorem.
\end{proof}

We introduce a closed expression for the operator which acts on $\cP$ as in theorem \ref{pxFischer}. This allows it to be defined on a function space larger than the polynomials.

\begin{definition}
\label{defpx}
The $q$-deformed Dirac operator is formally given by
\begin{eqnarray*}
\upx^q&=&\frac{1}{\ux}[\ux\upx]_q\\
&=&\frac{1}{\ux}[\mE+\Gamma]_q=\frac{\ux}{\ux^2}([\mE]_q+q^\mE[\Gamma]_q)\\
&=&-\uxi\left(\partial_r^q+\frac{1}{r}q^{r\partial_r}[\Gamma]_q\right).
\end{eqnarray*}
\end{definition}

\begin{remark}
It is important to note that $\mE$ and $\Gamma$ commute, so $q^{\mE+\Gamma}=q^{\mE}q^{\Gamma}$. This operator is clearly defined everywhere on functions in the space $\cP\otimes J$ with $J$ functions of $r$ on $\mR^+$.  This corresponds to the spaces mostly used in quantum Euclidean space (see e.g. \cite{MR1288667, MR1239953}).
\end{remark}

The operator $q^\mE$ can always be defined on $f(\ux)$ if $q\ux$ is in the domain of $f$. It is harder to define $q^\Gamma$. It can be defined locally on analytic functions. The Cauchy-Kowalewskaya theorem on the system
\begin{eqnarray*}
\partial_ug(\ux,u)=\Gamma_{\ux}g(\ux,u)& &g(\ux,0)=f(\ux),
\end{eqnarray*}

states that $g(\ux,u)$ is analytical when $f(\ux)$ is. Because $g$ is analytical,
\begin{eqnarray*}
q^\Gamma f(\ux)&=&\sum_{j=0}^\infty \frac{(\ln q)^j}{j!}\Gamma_{\ux}^jg(\ux,0)\\
&=&\sum_{j=0}^\infty \frac{(\ln q)^j}{j!}\left(\partial_u^jg\right)(\ux,0)\\
&=&g(\ux,\ln q).
\end{eqnarray*}

\begin{remark}
We could also consider functions which are only defined on $\partial \mB(R_i)$ for some $R_i\in\mR^+$ and are analytical on these $(m-1)$-dimensional manifolds. The operator $\Gamma_{\ux}$ is elliptic on these manifolds.
\end{remark}

All the functions we will encounter in this paper are polynomials times radial functions which pose no problem.
When we take the case $m=1$ we find that 
\begin{eqnarray*}
\upx^q&=&\frac{1}{e_1x_1}[e_1x_1(-e_1\partial_{x_1})]_q\\
&=&-\frac{e_1}{x_1}[x_1\partial_{x_1}]_q\\
&=&-e_1\partial_{x_1}^q,
\end{eqnarray*}

so the one dimensional case is a special case of this theory.

\begin{theorem}
The operator $\upx^q$ in definition \ref{defpx} is the unique linear operator on $\cP$ satisfying axioms $(A1)-(A4)$.
\end{theorem}
\begin{proof}
Definition \ref{defpx} and equations \eqref{1dirac1} and \eqref{1dirac2} imply that $\upx^q$ satisfies the properties in theorem \ref{pxFischer}, so it is unique. We only need to show that $\upx^q$ satisfies the axioms $(A1)-(A4)$ to prove the existence. Axiom $(A1)$ is trivial, axiom $(A2)$ follows from formula \eqref{1Gamma2},
\begin{eqnarray*}
\upx^q\ux^2&=&\frac{1}{\ux}[\mE+\Gamma]_q\ux^2\\
&=&\ux^2\frac{1}{\ux}[\mE+2+\Gamma]_q\\
&=&\ux^2\frac{1}{\ux}(q+1+q^2[\mE+\Gamma]_q)\\
&=&q^2\ux^2\upx^q+(q+1)\ux.
\end{eqnarray*}

To prove axiom $(A3)$ we calculate
\begin{eqnarray*}
(\upx^q)^2\ux^{2l}M_k&=&[2l-2+2k+m]_q[2l]_q\ux^{2l-2}M_k\\
(\upx^q)^2\ux^{2l+1}M_{k-1}&=&[2l]_q[2l+2k-2+m]_q\ux^{2l-1}M_{k-1}.
\end{eqnarray*}

Since every scalar spherical harmonic can be decomposed as $H_k=M_k+\ux M_{k-1}$, we find that
\begin{eqnarray}
\label{laplFischer}
(\upx^q)^2\ux^{2l}H_k&=&[2l]_q[2l-2+2k+m]_q\ux^{2l-2}H_k.
\end{eqnarray}
Since the set $\{\ux^{2l}H_k\}$ spans all scalar polynomials (lemma \ref{fischerdecomp}), $(\upx^q)^2$ acting on every scalar polynomial is scalar. Axiom $(A4)$ follows immediately from the definition.
\end{proof}

The operator in definition \ref{defpx} is of the form 
\begin{eqnarray}
\label{partieleafg}
\upx^q=-e_i\sum_{i=1}^mD_i
\end{eqnarray}
where $D_i$ are scalar operators. Because of lemma \ref{fischerdecomp} it suffices to calculate the action on scalar polynomials $\ux^{2l}H_k$, with decomposition $H_k=M_k+\ux M_{k-1}$,
\begin{eqnarray*}
\upx^q \ux^{2l}H_k&=&[2l]_q\ux^{2l-1}H_k+q^{2l}\ux^{2l}[m+2k-2]_qM_{k-1}\\
&=&[2l]_q\ux^{2l-1}H_k+q^{2l}\ux^{2l}\frac{[m+2k-2]_q}{m+2k-2}\upx H_{k}
\end{eqnarray*}

which clearly is a vector. Since, by axiom $(A3)$, the square of $\upx^q$ is a scalar operator,
\begin{eqnarray*}
(\upx^q)^2&=&\sum_{i,j=1}^me_ie_jD_iD_j\\
&=&-\sum_{i=1}^mD_i^2+\sum_{i<j}e_ie_j(D_iD_j-D_jD_i)\\
\end{eqnarray*}

is scalar. Since the set $\{e_ie_j,i<j\}$ is linearly independent, the operators $D_i$ must all commute. We will call them the $q$-partial derivatives. The Dirac operator $\upx$ is invariant under the action of Spin$(m)$, the universal cover of $SO(m)$. How the spin group can be realized in Clifford analysis can be found in \cite{MR697564} and  \cite{MR1169463}. Because multiplication with $\ux$ is also Spin$(m)$-invariant, we find that $\upx^q$, as defined by definition \ref{defpx} is also Spin$(m)$-invariant.

\begin{lemma}
The $q$-Dirac operator in definition \ref{defpx} satisfies
\label{pxx}
\begin{eqnarray*}
\upx^q\ux&=&[\mE-\Gamma+m]_q.
\end{eqnarray*}
\end{lemma}
\begin{proof}
We use commutation rule (\ref{1Gamma1}) to calculate
\begin{eqnarray*}
\upx^q\ux=\frac{1}{\ux}[\mE+\Gamma]_q\ux=[\mE+1+m-1-\Gamma]_q.
\end{eqnarray*}
\end{proof}

\begin{lemma}
\label{pxr}
For $f$ a scalar function of $r$, we have the following Leibniz rule
\begin{eqnarray*}
\upx^qf(r)=\upx^q(f(r))+f(qr)\upx^q=\frac{f(qr)-f(r)}{(q-1)\ux}+f(qr)\upx^q.
\end{eqnarray*}
\end{lemma}
\begin{proof}
Because $\Gamma$ commutes with $r$ we find
\begin{eqnarray*}
\upx^qf(r)&=&\frac{1}{\ux}([\mE]_qf(r)+q^\mE f(r)[\Gamma]_q)\\
&=&\frac{1}{\ux}([\mE]_qf(r)- f(qr)[\mE]_q)+\frac{1}{\ux}f(qr)[\mE]_q+\frac{1}{\ux}f(qr)q^\mE[\Gamma]_q\\
&=&\frac{1}{\ux}\frac{1}{q-1}(f(qr)q^\mE-f(r)- f(qr)q^\mE+f(qr))+f(qr)\upx^q\\
&=&\frac{f(qr)-f(r)}{(q-1)\ux}+f(qr)\upx^q.
\end{eqnarray*}
\end{proof}

\subsection{The $q$-Laplace operator}
As in the undeformed case we define the $q$-Laplace operator as minus the square of the $q$-Dirac operator.
\begin{definition}
\label{defLapl}
The $q$-deformed Laplace operator on analytic functions is given by
\begin{eqnarray*}
\Delta_q=-(\upx^q)^2.
\end{eqnarray*}
\end{definition}

Because $\upx^q\upx^q(M_k+\ux M_{k-1})=0$ we find that the spherical harmonics are the polynomial null solutions of the $q$-Laplace operator. The undeformed Laplace operator can be decomposed into its radial and angular part,
\begin{eqnarray*}
r^2\Delta&=&\mE(m-2+\mE)+\Gamma(m-2-\Gamma).
\end{eqnarray*}

The angular part is the Laplace-Beltrami operator
\begin{eqnarray}
\label{defLB}
\Delta_{LB}&=&\Gamma(m-2-\Gamma),
\end{eqnarray}
which is clearly scalar although it is defined here using the Clifford valued Gamma operator. We will also derive such a decomposition for the $q$-Laplace operator. It turns out that the angular part of the $q$-Laplace operator will be given by
\begin{definition}
\label{defqLB}
The $q$-Laplace-Beltrami operator on analytic functions is defined as
\begin{eqnarray*}
\Delta_{LB}^q&=&[\Gamma]_q[m-2-\Gamma]_q.
\end{eqnarray*}
\end{definition}

This operator is scalar, which is not obvious at first sight. This is a consequence of the decomposition of the $q$-Laplace operator in theorem \ref{decompLapl}. Property \eqref{1Gamma2} of the Gamma operator implies that the $q$-Laplace Beltrami operator commutes with radial functions. 
\begin{theorem}
\label{decompLapl}
The $q$-Laplace operator can be decomposed as
\begin{eqnarray*}
\Delta_q&=&q^{m-1}(\partial_r^q)^2+[m-1]_q\frac{1}{r}\partial_r^q+\frac{1}{r^2}q^\mE\Delta_{LB}^q,\\
r^2\Delta_q&=&[\mE]_q[m-2+\mE]_q+q^\mE[\Gamma]_q[m-2-\Gamma]_q.
\end{eqnarray*}
\end{theorem}

\begin{proof}
We calculate using definition \ref{defpx}, lemma \ref{pxx} and formula \eqref{1Gamma1}
\begin{eqnarray*}
r^2\Delta_q&=&\ux[\mE+\Gamma]_q[\mE-\Gamma+m]_q\frac{1}{\ux}\\
&=&[\mE+m-2-\Gamma]_q[\mE+\Gamma]_q\\
&=&\left([\mE+m-2]_q+q^{\mE+m-2}[-\Gamma]_q\right)\left([\mE]_q+q^\mE[\Gamma]_q \right)\\
&=&[\mE+m-2]_q[\mE]_q+q^\mE \left([\Gamma]_q[\mE+m-2-\Gamma]_q+q^{m-2}[-\Gamma]_q[\mE]_q\right)\\
&=&[\mE+m-2]_q[\mE]_q+q^\mE [\Gamma]_q\left([\mE+m-2-\Gamma]_q+q^{m-2}\frac{q^{-\Gamma}(1-q^\Gamma)}{q^\Gamma-1}[\mE]_q\right)\\
&=&[\mE+m-2]_q[\mE]_q+q^\mE [\Gamma]_q[m-2-\Gamma]_q.
\end{eqnarray*}

This leads to the second expression, the first one can be found from
\begin{eqnarray*}
\frac{1}{r^2}[\mE]_q[m-2+\mE]_q&=&\frac{1}{r}\partial_r^q([m-2]_q+q^{m-2}r\partial_r^q)\\
&=&[m-2]_q\frac{1}{r}\partial_r^q+q^{m-2}\frac{1}{r}\partial_r^q+q^{m-1}(\partial_r^q)^2.
\end{eqnarray*}
\end{proof}

\begin{remark}
In \cite{DBS4} a theory of Clifford analysis in superspace was developed by constructing a Dirac operator which satisfies $\px^2=\Delta$ with $\Delta$ the well-known orthosymplectic super Laplace operator. Using definition \ref{defpx} we can also construct a theory of $q$-deformed Clifford and harmonic analysis in superspace.
\end{remark}

The decomposition of the $q$-Laplace operator in theorem \ref{decompLapl} can be used to calculate the action on the product of a radial function and a spherical harmonic.

\begin{lemma}
\label{laplHk}
For $f$ a function of $r$ and $H_k$ a spherical harmonic of degree $k$, the following holds
\begin{eqnarray*}
\Delta_qf(r)H_k&=&H_k\left[q^{m-1+2k}(\partial_r^q)^2+[m-1+2k]_q\frac{1}{r}\partial_r^q\right]f(r).
\end{eqnarray*}
\end{lemma}
\begin{proof}
Since $\Delta_q H_k=0$, theorem \ref{decompLapl} yields 
\begin{eqnarray*}
q^\mE\Delta_{LB}H_k=-[k]_q[m-2+k]_qH_k.
\end{eqnarray*}

We use this to calculate
\begin{eqnarray*}
\Delta_qf(r)H_k&=&H_k\left[\frac{1}{r^2}[\mE+k]_q[m-2+\mE+k]_q-\frac{1}{r^2}q^{\mE}[k]_q[m-2+k]_q\right]f(r)\\
&=&H_k\frac{1}{r^2}\left[[\mE]_q[m-2+\mE+k]_q+q^\mE[k]_q[m-2+\mE+k]_q-q^{\mE}[k]_q[m-2+k]_q\right]f(r)\\
&=&H_k\frac{1}{r^2}\left[[\mE]_q[m-2+\mE+k]_q+q^\mE[k]_qq^{m-2+k}[\mE]_q\right]f(r)\\
&=&H_k\frac{1}{r^2}[\mE]_q[m-2+\mE+2k]_qf(r).\\
\end{eqnarray*}

This is the usual action of $\Delta_q$ on $f(r)$, with substitution $m\to m+2k$.
\end{proof}

It is inelegant that scalar operators like the $q$-deformation of the Laplace and Laplace-Beltrami operator are defined only using Clifford algebras. Therefore we derive purely scalar expressions for $\Delta_{LB}^q$ and $\Delta_q$.
\begin{lemma}
\label{LBqLB}
The $q$-Laplace-Beltrami on analytic functions is given by
\begin{eqnarray*}
\Delta_{LB}^q&=&[\frac{m}{2}-1-\sqrt{(\frac{m}{2}-1)^2-\Delta_{LB}}]_q[\frac{m}{2}-1+\sqrt{(\frac{m}{2}-1)^2-\Delta_{LB}}]_q.
\end{eqnarray*}
\end{lemma}
\begin{proof}
It is not a priori clear that the right hand side is well defined. In the case $q=1$ we find
\begin{eqnarray*}
\left(\frac{m}{2}-1\right)^2-\left(\left(\frac{m}{2}-1\right)^2-\Delta_{LB}\right),
\end{eqnarray*}

so there does not really appear a square root of the Laplace-Beltrami operator, which would be ill-defined. The same thing happens in the $q$-deformed case. The right hand side is defined by a series expansion, so it is equal to the series expansion of
\begin{eqnarray*}
\frac{q^{m-2}-q^{\frac{m}{2}-1}2\cosh(\ln q\sqrt{(\frac{m}{2}-1)^2-\Delta_{LB}})+1}{(q-1)^2}.
\end{eqnarray*}

Using equation (\ref{defLB}), we calculate
\begin{eqnarray*}
\cosh(\ln q\sqrt{(\frac{m}{2}-1)^2-\Delta_{LB}})&=&\sum_{l=0}^\infty\frac{(\ln q\sqrt{(\Gamma-\frac{m}{2}+1)^2})^{2l}}{(2l)!}\\
&=&\sum_{l=0}^\infty\frac{(\ln q(\Gamma-\frac{m}{2}+1))^{2l}}{(2l)!}\\
&=&\cosh(\ln q(\Gamma-\frac{m}{2}+1)).
\end{eqnarray*}

This means the expression on the right hand side is equal to
\begin{eqnarray*}
\frac{q^{m-2}-q^{\frac{m}{2}-1}(q^{\Gamma-\frac{m}{2}+1}+q^{\frac{m}{2}-1-\Gamma})+1}{(q-1)^2}&=&\frac{q^{m-2}-q^\Gamma-q^{m-2-\Gamma}+1}{(q-1)^2},
\end{eqnarray*}

which is the $q$-Laplace-Beltrami operator in definition \ref{defqLB}.
\end{proof}

Similarly we can prove the following scalar expressions for the $q$-Laplace operator.
\begin{theorem}
\label{scalqLapl}
The $q$-Laplace operator on analytical functions is given by
\begin{eqnarray*}
\Delta_q&=&\frac{1}{r^2}\left[\mE+\frac{m}{2}-1+\sqrt{(\mE+\frac{m}{2}-1)^2-r^2\Delta}\right]_q\left[\mE+\frac{m}{2}-1-\sqrt{(\mE+\frac{m}{2}-1)^2-r^2\Delta}\right]_q\\
&=&\frac{1}{r^2}\left[\mE+\frac{m}{2}-1+\sqrt{(\frac{m}{2}-1)^2-\Delta_{LB}}\right]_q\left[\mE+\frac{m}{2}-1-\sqrt{(\frac{m}{2}-1)^2-\Delta_{LB}}\right]_q.
\end{eqnarray*}
\end{theorem}

As we will see the $q$-Laplace operator is related to a fundamental object in quantum Euclidean space, without connection to Clifford analysis. It is remarkable that it is defined more elegantly using Clifford algebras (which disappear in the resulting operator) in theorem \ref{decompLapl}, than without Clifford algebras, in theorem \ref{scalqLapl}.

\subsection{A $q$-deformed version of $\mathfrak{sl}_2(\mR)$ and $\mathfrak{osp}(1|2)$}

In classical harmonic analysis the $SO(m)$-invariant operators $r^2/2$, $\Delta/2$ and $\mE+\frac{m}{2}$ generate the Lie algebra $\mathfrak{sl}_2(\mR)$, see \cite{MR0986027}. These operators also generate as an associative algebra the universal enveloping algebra of $\mathfrak{sl}_2(\mR)$. By $q$-deforming this to $U_q(\mathfrak{sl}_2(\mR))$ we take one of the two dually related ways to $q$-deform a Lie-algebra. We define
\begin{equation}
\label{defE}
E=\frac{q+1}{4}\left(\upx^q\ux+q^2\ux\upx^q\right)=\frac{q+1}{4}\left([\mE+m-\Gamma]_q+q^2[\mE+\Gamma]_q\right).
\end{equation}

This operator is scalar, see lemma \ref{scalar}. We will use the notations $\{ A,B\}_c=AB+cBA$ and $[A,B]_c=AB-cBA$.  Rewriting equation (\ref{EE}) and a straightforward calculation lead to
\begin{eqnarray*}
\left[ \Delta_q/2, r^2/2 \right]_{q^4} &=& E\\
\left[ \Delta_q/2, E\right]_{q^2} &=&  \frac{[4]_q[2]_q}{4}\Delta_q/2 \\
\left[ E,r^2/2 \right]_{q^2} &=& \frac{[4]_q[2]_q}{4} r^2/2.
\end{eqnarray*}

So $r^2/2$, $\Delta_q/2$ and $E$ form a $q$-deformed version of $\mathfrak{sl}_2(\mR)$. This corresponds to the $\mathfrak{su}(1|1)_q$ quantum algebra in \cite{MR1146380}, which is defined by operators $L_1,L_{-1}$ and $L_0$, satisfying
\begin{eqnarray*}
q^{-2}L_1L_{-1}-q^2L_1{-1}L_1 &=& q\frac{[4]_q}{[2]_q}L_0\\
q^{-1}L_1L_0-qL_0L_1 &=&  L_1\\
q^{-1}L_0L_{-1}-qL_{-1}L_0 &=&  L_{-1}.
\end{eqnarray*}

This algebra is obtained from the identification $L_1=\frac{q\Delta_q}{[2]_q}$, $L_{-1}=\frac{qr^2}{[2]_q}$ and $L_0=\frac{4q}{[4]_q[2]_q}E$. Now we prove that these generators of $\mathfrak{sl}_2(\mR)_q$ are still $SO(m)$-invariant. Therefore, a deformation of the Howe dual pair $\left(SO(m),\mathfrak{sl}_2(\mR)\right)$ to $\left(SO(m),\mathfrak{sl}_2(\mR)_q\right)$ is obtained.

\begin{lemma}
The operators $r^2$, $\Delta_q$ and $E$ are $SO(m)$-invariant.
\end{lemma}

\begin{proof}
Since the undeformed operators are $SO(m)$-invariant, we find that $r^2$ is $SO(m)$-invariant and using theorem \ref{scalqLapl} that $\Delta_q$ is $SO(m)$-invariant. Because $E$ can be written as the $q^4$-commutator of $r^2$ and $\Delta_q$ it is also invariant.
\end{proof}

The module $\oplus_jr^{2j}H_k$ forms a lowest weight module for the representation of $\mathfrak{sl}_2(\mR)$ given by the action of $\Delta,$ $r^2$ and $\mE+m/2$. The lowest weight vector is $H_k$ with lowest weight $m/2+k$. The action of $\Delta_q, r^2$ and $E$ has the same structure but with $q$-deformed coefficients, so we also obtain a lowest weight module for $\mathfrak{sl}_2(\mR)_q$.

We can consider a larger algebra than $\mathfrak{sl}_2$, generated by $\upx$ and $\ux$. Then we find the Lie superalgebra $\mathfrak{osp}(1|2)$. Here we give the $q$-deformed commutation rules of the algebra generated by $\upx^q$ and $\ux$
\[
\{\ux,\ux\} = -2 r^2\qquad \{\upx^q,\ux\}_{q^2} =\frac{q+1}{2} E\qquad \{\upx^q,\upx^q\} = -2 \Delta_q
\]

\label{qcommu}
and
\[
\begin{array}{lllllll}
\left[\ux, r^2\right] &=&0 &\;& \left[\upx^q, r^2\right]_{q^2} &=&-(q+1)\ux\\
\left[\Delta_q, \ux \right]_{q^2} &=&-(q+1)\upx^q &\;& \left[\upx^q, \Delta_q\right] &=&0\\
\left[E, \ux\right]_{q^2} &=&\frac{(q+1)^2}{4}\ux-q^2(1-q^2)r^2\upx^q &\;& \left[\upx^q,E\right]_{q^2} &=&\frac{(q+1)^2}{4}\upx^q-q^2(1-q^2)\ux\Delta_q.\\
\end{array}
\]

As an illustration we calculate $\left[\Delta_q, \ux \right]_{q^2} $, using definition \ref{defpx}, lemma \ref{pxx}, the fact that $\mE$ and $\Gamma$ commute and axiom $(A2)$
\begin{eqnarray*}
(\upx^q)^2\ux&=&\frac{1}{\ux}\ux\upx^q\upx^q\ux=\frac{1}{\ux}\upx^q\ux\ux\upx^q\\
&=&\frac{1}{\ux}((q+1)\ux+q^2\ux^2\upx^q)\upx^q\\
&=&(q+1)\upx^q+q^2\ux(\upx^q)^2.\\
\end{eqnarray*}

Since $\upx^q$ is Spin$(m)$-invariant we also obtain the Howe dual pair $\left(\mbox{Spin}(m),\mathfrak{osp}(1|2)_q\right)$.

\section{$q$-analogues of the radial Schr\"odinger equation}

In \cite{MR1288667, MR1239953} the Schr\"odinger equation of the harmonic oscillator in the $m$-dimensional quantum Euclidean space was studied. The symmetry group of the construction is $SO_q(m)$, see \cite{MR1097174}. The Hopf algebra $Fun_q(SO(m))$ of functions on $SO_q(m)$ are power series in $T_{ij}$, with $T_{ij}(g)$ the matrix of the fundamental representation for $g\in SO(m)$. They satisfy $TCT^T=C$, for the metric $C$ and have commutation relations determined by the braid matrix $\hat{R}$, $\hat{R}^{ij}_{kl}T^k_s T^l_p=T^j_lT^i_k\hat{R}^{kl}_{sp}$. In the undeformed case $\hat{R}_{kl}^{ij}=\delta_{ik}\delta_{jl}$. The braid matrix can be written using projection operators as
\begin{eqnarray*}
\hat{R}&=&q\cP_S-q^{-1}\cP_A+q^{1-m}\cP_1.
\end{eqnarray*}

The braid matrix is connected to the metric by the relation $(\cP_1)^{ij}_{kl}=\frac{C^{ij}C_{kl}}{C^{pq}C_{pq}}$. The commutation relations for the variables and the derivatives are given by $(\cP_A)_{kl}^{ij}x^kx^l=0$ and $(\cP_A)^{ij}_{kl}\partial^k\partial^l=0$. The action of the derivatives is given by the Leibniz rule
\begin{eqnarray*}
\partial^ix^j&=&C^{ij} +q\hat{R}^{ij}_{kl}x^k\partial^l.
\end{eqnarray*}

The metric is used to define the generalized norm squared $x^2=x\cdot x=x^iC_{ij}x^j$ and the Laplace operator $\partial\cdot\partial =\partial^iC_{ij}\partial^j$, they are clearly $SO_q(m)$-invariant. The function space considered is freely generated by the $x^k$ modulo the $\cP_A$-commutation relations. The center is generated by $1$ and $x^2$ (see \cite{MR1288667, MR1239953, MR1097174}).

This allows to construct a $q$-deformed Hamiltonian with the corresponding Schr\"odinger equation
\begin{eqnarray}
\label{qgroupharm}
H\,\Psi=\left[-q^m\partial\cdot\partial+x\cdot x\right]\Psi&=&E\,\Psi,
\end{eqnarray}
which has an $SO_q(m)$ symmetry. In \cite{MR1288667} this equation was first solved by constructing creation and annihilation operators. Then it was shown that this equation could also be solved using an ansatz of the form
\begin{eqnarray*}
\Psi&=&S^I_k\, g(x^2),
\end{eqnarray*}

where $S^I_k$ is of degree $k$ and satisfies $\partial\cdot\partial S^I_k=0$, so it replaces the notion of a spherical harmonic. The Schr\"odinger equation (\ref{qgroupharm}) then led to the following equation (we use an unimportant different normalization of the energy)
\begin{eqnarray}
\label{qgroupharm2}
\left[-q^{m+2k}x^2(\partial_{x^2}^{q^2})^2-[\frac{m}{2}+k]_{q^2}\partial_{x^2}^{q^2} +\frac{x^2}{(q+1)^2}\right]g(x^2)&=&\frac{E}{q+1}g(x^2).
\end{eqnarray}
In this equation $x^2$ can be treated as a normal variable, so we take $r^2=x^2$ and $g$ has to satisfy a $q$-difference equation. By substituting $g(x^2)=f(r)$ and calculating
\begin{eqnarray*}
\partial_{x^2}^{q^2}g(x^2)&=&\frac{g(q^2x^2)-g(x^2)}{(q^2-1)x^2}\\
&=&\frac{1}{(q+1)r}\frac{f(qr)-f(r)}{(q-1)r}
\end{eqnarray*}

we find that equation (\ref{qgroupharm2}) leads to
\begin{eqnarray}
\label{Schro1}
\frac{1}{q+1}\left[-q^{m+2k-1}(\partial_r^q)^2-[m+2k-1]_q\frac{1}{r}\partial_{r}^{q} +r^2\right]f(r)&=&Ef(r).
\end{eqnarray}
This equation was studied in \cite{MR1447846} and \cite{Papp}. With the $q$-deformed Laplace operator in definition \ref{defLapl} it is possible to put this equation into a Schr\"odinger equation completely determined by $q$-analysis, without quantum variables. By lemma \ref{laplHk}, the equation (\ref{Schro1}) for $f(r)$ is equivalent with
\begin{eqnarray}
\label{Schro2}
\frac{1}{q+1}\left[-\Delta_q+r^2\right]f(r)H_k&=&Ef(r)H_k,
\end{eqnarray}
for $H_k$ an arbitrary spherical harmonic. So, the entire quantum system in $q$-Euclidean space can be replaced by the $q$-Schr\"odinger equation in undeformed space,
\begin{eqnarray*}
\frac{1}{q+1}\left[-\Delta_q+r^2\right]\Psi(\ux)&=&E\,\Psi(\ux).
\end{eqnarray*}

The dimension of the space of spherical harmonics does not depend on $q$, so $\dim \cS^I_k=\dim  \cH_k$, see \cite{MR1288667, MR1239953, MR1097174, MR1380862, MR1408120}. This means the energy eigenvalues and multiplicities of more general Schr\"odinger equations $\frac{1}{q+1}\left[-\Delta_q+V(r)\right]\Psi(\ux)=E\Psi(\ux)$ are equal to those of the corresponding Schr\"odinger equations in quantum Euclidean space. This spectrum can be found using separation of variables and the results in \cite{MR2316312}. As an example we consider the free particle $\Delta_q \psi(\ux)=-l^2\psi(\ux)$ as in \cite{MR1380862}. The $q$-difference equation
\begin{eqnarray*}
\left[q^{m+2k-1}(\partial_r^q)^2+[m+2k-1]_q\frac{1}{r}\partial_{r}^{q} \right] f(r)&=&-l^2f(r)
\end{eqnarray*}

with $f$ an even function is solved by
\begin{eqnarray*}
f(r)&=&\sum_{n=0}^\infty \frac{(-1)^n}{\Gamma_Q(n+1)\Gamma_Q(\frac{m}{2}+k+n)}\left(\frac{lr}{q+1}\right)^{2n}.
\end{eqnarray*}

This corresponds to the $q$-Bessel functions introduced by Jackson,
\begin{eqnarray*}
\frac{H_k(\ux)}{r^{k+\frac{m}{2}-1}}J^q_{\frac{m}{2}+k-1}(lr).
\end{eqnarray*}

The odd case leads to the $q$-Neumann functions, see \cite{MR1380862}.

\label{Schro}

\section{$q$-integration}
\subsection{One dimensional case}
The following lemma about one dimensional $q$-integration follows from straightforward calculations.
\begin{lemma}
For $q<1$, $k\in\mN$ and $a,b,c\in\mR$, the following relations holds,
\label{changevar}
\label{intnorm}
\label{intqinv}
\begin{eqnarray*}
(i)& &\int_{a^k}^{b^k}d_{q^k}t^k\,f(t)=[k]_q\int_a^bd_qt\,f(t)t^{k-1}\\
(ii)& &\int_{a}^bd_qt\, f(ct)=\frac{1}{c}\int_{ca}^{cb}d_qt\, f(t)\\
(iii)& &\int_{0}^{q^{-k}a}d_qt\,f(t)=\int_{0}^ad_qt\,f(t)+(1-q)a\sum_{i=1}^kf(aq^{-i})q^{-i}.
\end{eqnarray*}
\end{lemma}

For the sequel we will need the $q$-integral of $e_Q(-t^2)$ with $Q=q^2$. First we need the following lemma.

\begin{lemma}
\label{nulpexp}
The zeroes of the exponentials defined in (\ref{exp1}) and (\ref{exp2}) are given by
\begin{eqnarray*}
E_{q}(\frac{q^{k+1}}{1-q})&=&0\quad \mbox{if } \,q>1\mbox{   and}\\
e_{q}(\frac{q^{-k}}{q-1})&=&0\quad \mbox{if }\, q<1
\end{eqnarray*}

for $k\in\mN$.
\end{lemma}
\begin{proof}
We start from the $q$-difference property (\ref{aflexp}) of the $q$-exponential
\begin{eqnarray*}
E_q(qu)&=&\left(1+(q-1)u\right) E_q(u).
\end{eqnarray*}

This implies that $E_q(qu)=0$ if and only if either $E_q(u)=0$ or $u=\frac{1}{1-q}$. So we obtain $E_q(\frac{q^{k+1}}{1-q})=0$ for all $k\in\mN$. We still have to prove that these are the only possible zeroes. If we assume $E_q(t)=0$ with $t\not= \frac{q^{k+1}}{1-q}$ then this would imply, since $\lim_{j\to\infty}q^{-j}t=0$, that $E_q(0)=0$. This is not the case as formula \eqref{exp1} implies $E_q(0)=1$. The second claim can be found immediately by making the substitution $q\to q^{-1}$.
\end{proof}

Using lemma \ref{intqinv}, lemma \ref{nulpexp} and integral representation (\ref{intGamma}) we can calculate the $q$-analogue of\\$\int_{\mR}dt\,t^{\nu-1}\exp(-t^2)$. This result can also be found in \cite{MR0708496}.
\begin{lemma}
For $q<1$ and with $\lambda_Q=\sqrt{\frac{1}{1-Q}}$, the following holds
\label{gammagaus}
\begin{eqnarray*}
\int_{-\lambda_Q}^{\lambda_Q} d_qt\,t^{\nu-1}\,e_Q(-t^2)&=&\frac{2}{q+1}Q^{\frac{\nu}{2}}\Gamma_Q(\frac{\nu}{2}).
\end{eqnarray*}
\end{lemma}

\subsection{$q$-integration in $\mR^m$}
One dimensional $q$-integration is defined in equation (\ref{1Dqint}). The aim of this section is to generalize this concept to higher dimensions, corresponding to the $q$-deformation of the vector derivative. This means we want analogues of equation (\ref{partiele}). In classical Clifford analysis, these are given by Cauchy-type formulae in higher dimensions, see \cite{MR697564, MR1169463, MR1249888, MR1012510}. In this section we will always assume $q<1$. Before we define $q$-integration in $\mR^m$ we repeat a well-known fact about the $\Gamma$-operator. For $f$ and $g$ two Clifford valued differentiable functions
\begin{eqnarray*}
\int_{\mS^{m-1}}d\xi\,\,\,\overline{(\Gamma f)}g&=&\int_{\mS^{m-1}}d\xi\,\,\,\overline{f}(\Gamma g).
\end{eqnarray*}

This equation together with the series expansion of $q^\Gamma$ yields

\begin{lemma}
\label{gammaint}
For $f$ and $g$ two Clifford valued analytic functions,
\begin{eqnarray*}
\int_{\mS^{m-1}}d\xi\,\,\,\overline{(q^\Gamma f)}\,g&=&\int_{\mS^{m-1}}d\xi\,\,\,\overline{f}\,(q^\Gamma g).
\end{eqnarray*}
\end{lemma}

Now we define our $q$-integration on $\mR^m$. There have been made other approaches to generalize Jackson's $q$-integration to higher dimensions (see \cite{MR1239953, MR1408120, MR1303081, MR2035025}), but those are integrations over quantum variables, while we use a $q$-integration over commuting variables. One approach is based on Gaussian integration and the necessity for a Stokes theorem (\cite{MR1239953, MR1303081}).  Our approach is more closely related to integration over the quantum Euclidean sphere (\cite{MR1408120}), but as we will see, also satisfies Stokes theorem.
\begin{definition}
For every function $f$ on the ball with radius $R$, $\mB^m(R)$, for which the expression is finite, the $m$-dimensional $q$-integral is given by
\begin{eqnarray*}
\int_{\mB^m(R)}fd_qV(\ux)&=&\int_{\mS^{m-1}}d\xi\int_{0}^Rd_q r\, r^{m-1}f,
\end{eqnarray*}

with $d_qr$ the measure in (\ref{1Dqint}).
\end{definition}
We could also use the infinite Jackson $q$-integration (see \cite{MR0708496, MR1052153, JACKSON}) to construct $q$-integration on entire $\mR^m$ but we will not need it here.

\begin{remark}
The function $f$ only has to be defined on the spheres $\partial\mB^m(q^kR)$, $k\in\mN$ for this integral to be well defined. By considering all integrations on the balls $\mB^m(q^l)$ for $l\in\mZ$ we obtain a mapping of functions defined on $\{\partial\mB^m(q^k)|k\in\mZ\}\subset\mR^m$ into functions defined on the set of points $\{q^k|k\in\mZ\}$. 
\end{remark}

Applying lemma \ref{intnorm}$(ii)$ yields
\begin{eqnarray}
\label{intnorm2}
\int_{\mB^m(R)}d_qV(\ux)f(c\ux)&=&\frac{1}{c^m}\int_{\mB^m(cR)}d_qV(\ux)f(\ux).
\end{eqnarray}

Now we are ready to state and prove the Cauchy formula for the $q$-Dirac operator.
\begin{theorem}
\label{qCauchy}($q$-Cauchy formula)\\
For $f$ and $g$ two Clifford valued analytical functions on $\mB^m(R)$, the following relation holds,
\begin{eqnarray*}
\int_{\mB^m(R)}d_qV(\ux)\,\left[\overline{\left(q^{\Gamma}\upx^qf\right)}\,g(q\ux)-\overline{f}\,(\upx^qg)\right]&=&R^{m-2}\int_{\partial\mB^{m}(R)}d\xi\,\,\,\overline{f(\ux)}\,\ux \,g(\ux).
\end{eqnarray*}
\end{theorem}

\begin{proof}
First we use equations (\ref{partiele}), (\ref{Leibniz}) and (\ref{Leibniz2}) to calculate
\begin{eqnarray*}
\int_0^Rd_qr\,\partial_r^q\left[r^m\overline{f}\frac{1}{\ux}g\right]&=&-R^{m-1}\overline{f(R\uxi)}\,\uxi \,g(R\uxi)\\
\end{eqnarray*}
\[
=\int_0^Rd_qr\,[m]_qr^{m-1}\overline{f}\frac{1}{\ux}g+\int_0^Rd_qr\,q^mr^{m}\partial_r^q(\overline{f})\frac{1}{q\ux}g(q\ux)+\int_0^Rd_qr\,q^mr^{m}\overline{f}\partial_r^q(\frac{1}{\ux}g)\]\[
=\int_0^Rd_qr\,r^{m-1}\overline{f}\left( [m]_q+q^m[\mE]_q\right)\frac{1}{\ux}g+q^{m-1}\int_0^Rd_qr\,r^{m-1}([\mE]_q\overline{f})\frac{1}{\ux}g(q\ux).
\]

The above, lemma \ref{pxx} and lemma \ref{gammaint} lead to
\begin{eqnarray*}
\int_{\mB^m(R)}d_qV(\ux)\,\overline{f}\,(\upx^qg)&=&\int_{\mS^{m-1}}d\xi\int_{0}^Rd_qr\,r^{m-1}\overline{f}\left([m]_q+q^m[\mE]_q+[-\Gamma]_qq^{m+\mE}\right)\frac{1}{\ux}g\\
&=&-R^{m-1}\int_{\mS^{m-1}}d\xi\,\,\,\overline{f(R\uxi)}\,\uxi \,g(R\uxi)-q^{m-1}\int_{\mB^m(R)}d_qV(\ux)([\mE]_q\overline{f})\frac{1}{\ux}g(q\ux)\\
&+&q^{m}\int_{\mS^{m-1}}d\xi\int_{0}^Rd_qr\,r^{m-1}\overline{([-\Gamma]_qf)}\frac{1}{q\ux}g(q\ux)\\
&=&-R^{m-2}\int_{\partial\mB^{m}(R)}d\xi\,\,\,\overline{f(\ux)}\,\ux \,g(\ux)+q^{m-1}\int_{\mB^m(R)}d_qV(\ux)\overline{\frac{1}{\ux}\left(q^{-\Gamma}([\mE+\Gamma]_q)f\right)}g(q\ux)\\
&=&-R^{m-2}\int_{\partial\mB^{m}(R)}d\xi\,\,\,\overline{f(\ux)}\,\ux \,g(\ux)+\int_{\mB^m(R)}d_qV(\ux)\,\overline{\left(q^{\Gamma}(\upx^q)f\right)}\,g(q\ux).
\end{eqnarray*}
This concludes the proof.
\end{proof}
As a special case of this theorem we obtain the generalization of formula \eqref{partiele} to the $m$-dimensional case.
\begin{corollary}
\label{corint}
For $g$ a Clifford valued analytical function on $\mB^m(R)$, the following Cauchy-formula holds,
\begin{eqnarray*}
\int_{\mB^m(R)}d_qV(\ux)\,(\upx^qg)&=&-R^{m-2}\int_{\partial\mB^{m}(R)}d\xi\,\,\,\,\ux \,g(\ux).
\end{eqnarray*}
\end{corollary}

When we take $g$ scalar, the formula in this corollary falls apart into formulas for the $q$-partial derivatives $D_i$ in formula (\ref{partieleafg}). In particular, for a function which vanishes on $\partial \mB^m(R)$, corollary \ref{corint} implies
\begin{eqnarray*}
\int_{\mB^m(R)}d_qV(\ux)\,D_ig&=&0.
\end{eqnarray*}

This shows the link with the Gaussian integration method in \cite{MR1239953} and \cite{MR1303081}. The $q$-partial derivatives $D_i$ take the place of the derivatives with respect to the quantum variables.

The term $q^\Gamma$ which appears in theorem \ref{qCauchy} is dropped when we consider the Laplace operator.
\begin{corollary}
For $f$ and $g$ two Clifford valued analytical functions on $\mB^m(R)$ with $g=0=\upx^qg$ on $\partial \mB^m(R)$,
\begin{eqnarray*}
\int_{\mB^m(R)}d_qV(\ux)\,f(q\ux)\,(\Delta_qg)&=&\int_{\mB^m(R)}d_qV(\ux)\,(\Delta_qf)\,g(q\ux).
\end{eqnarray*}
\end{corollary}
\begin{proof}
We start by putting $f=\overline{h}$ and using theorem \ref{qCauchy},
\begin{eqnarray*}
\int_{\mB^m(R)}d_qV(\ux)\,\overline{h(q\ux)}\,(\upx^q\upx^qg)&=&-\int_{\mB^m(R)}d_qV(\ux)\,\overline{q^\Gamma\upx^qq^\mE f(\ux)}\,(q^\mE\upx^qg),
\end{eqnarray*}

where the surface term vanished because $(\upx^qg)=0$ on the boundary. Using formula (\ref{intnorm2}), lemma \ref{intnorm}$(iii)$ with $(\upx^qg)(R\uxi)=0$ and theorem \ref{qCauchy} then leads to
\begin{eqnarray*}
\int_{\mB^m(R)}d_qV(\ux)\,\overline{h(q\ux)}\,(\upx^q\upx^qg)&=&-q\int_{\mB^m(R)}d_qV(\ux)\,q^\mE\left[\overline{q^\Gamma\upx^q h(\ux)}\,(\upx^qg)\right]\\
&=&-\frac{1}{q^{m-1}}\int_{\mB^m(qR)}d_qV(\ux)\,\overline{q^\Gamma\upx^q h(\ux)}\,(\upx^qg)\\
&=&-\frac{1}{q^{m-1}}\int_{\mB^m(R)}d_qV(\ux)\,\overline{q^\Gamma\upx^q h(\ux)}\,(\upx^qg)\\
&=&\frac{1}{q^{m-1}}\int_{\mB^m(R)}d_qV(\ux)\,\overline{q^\Gamma\upx^qq^\Gamma\upx^q h(\ux)}\,g(q\ux)\\
&=&\frac{1}{q^{m-1}}\int_{\mB^m(R)}d_qV(\ux)\,\overline{q^\Gamma q^{m-1-\Gamma}\upx^q\upx^q h(\ux)}\,g(q\ux).
\end{eqnarray*}

The surface term in the $q$-Cauchy theorem was again zero because $g=0$. The $q$-Laplace operator is scalar, so $\overline{\Delta_q}=\Delta_q$ and the proposed formula is obtained.
\end{proof}

\section{Hermite polynomials}
\subsection{One dimensional case}
\label{sectionHerm}
A lot of approaches have been used to study $q$-deformed versions of the Hermite polynomials, see e.g. \cite{MR0708496, MR1448659, MR0599807, MR2392899}. Because of the different definitions and normalizations in the literature we give a short overview of the $q$-Hermite polynomials. We choose a normalization such that $\lim_{q\to 1}H_k^q(t)=H_k(t)$, with $H_k$ the classical Hermite polynomials. The starting point is the $q$-Hermite's equation of Exton, see \cite{MR0599807}. This leads to a recursion relation, which is mostly used to define $q$-Hermite polynomials. We will also calculate the creation and annihilation operators and derive an orthogonality property. Most of these results can be found in \cite{MR0708496}.

\begin{definition}
\label{defherm}
The $q$-Hermite polynomial $H_k^q$ is the polynomial of the form
\begin{eqnarray*}
H^q_k(t)&=&\sum_{j=0}^{\lfloor k/2\rfloor}a_k^jt^{k-2j},
\end{eqnarray*}

with $a_k^0=(q+1)^k$, which is an eigenvector of the $q$-Hermite's equation
\begin{eqnarray*}
[(\partial_t^q)^2-(q+1)t\partial_t^q]f(t)&=&-(q+1)\lambda f(q\, t).
\end{eqnarray*}
\end{definition}

From the definition we immediately find that the eigenvalues are 
\[
\lambda_k=[k]_qq^{-k}.
\]

The exact form of $H_k^q$ is
\begin{equation}
\label{vormqherm}
H_k^q(t)=\sum_{j=0}^{\lfloor k/2\rfloor}(q+1)^{k-j}\frac{[k]_q!}{[k-2j]_q!}\frac{t^{k-2j}}{[-2j]_q[-2j+2]_q\cdots [-2]_q}.
\end{equation}

Taking the limit $q\to1$ we find $H_k(t)=\sum_{j=0}^{\lfloor k/2\rfloor}(-1)^j2^{k-2j}\frac{k!}{(k-2j)!j!}t^{k-2j}$. Now we show the recursion formula and the annihilation operator. The simplest way to prove these is by considering the coefficients.

\begin{theorem}
\label{hermrec}
The following recursion formula holds for the polynomials introduced in definition \ref{defherm},
\begin{eqnarray*}
(i)\,\,H^q_{k+1}&=&(q+1)tH^q_k-(q+1)[k]_q q^{k+1}H^q_{k-1}
\end{eqnarray*}

when $k>0$. The annihilation operator for the $q$-hermite polynomials is $\partial^q_t$,
\begin{eqnarray*}
(ii)\,\,\partial_t^qH^q_k(t)&=&(q+1)[k]_qH_{k-1}^q(t).
\end{eqnarray*}
\end{theorem}

In the classical case the creation operator can be obtained from either the combination of the annihilation operator and the recursion formula or the combination of the annihilation operator and the Hermite's equation. In the $q$-deformed case these two approaches lead to different creation operators.
\begin{theorem}
\label{hermcrea}
For $H_k^q$ as defined in definition \ref{defherm} the following relations hold for $k>0$,
\begin{eqnarray*}
(i)\,\,H^q_k(t)&=&((q+1)t-q^k\partial_t^q)H^q_{k-1}(t)
\end{eqnarray*}

and
\begin{eqnarray*}
(ii)\,\,H^q_k(qt)&=&q^{k}((q+1)t-\partial_t^q)H^q_{k-1}(t).
\end{eqnarray*}
\end{theorem}

For our purpose we use the following $q$-exponential based on formula (\ref{exp2}) with $Q=q^2$,
\begin{eqnarray}
\label{defexp}
e_Q(u)&=&\sum_{j=0}^\infty q^{j(j-1)}\frac{u^{j}}{[j]_{q^2}!}.
\end{eqnarray} 
This exponential satisfies $\partial_t^q\left[ e_Q(-t^2)\right]=-(q+1)te_Q(-q^2t^2)$. Together with theorem \ref{hermcrea}$(ii)$ and Leibniz rule (\ref{Leibniz}) this yields
\begin{eqnarray}
\label{hermexp}
H_k^q(qt)\, e_Q(-q^2t^2)&=&-q^k\,\partial_t^q\,[H_{k-1}^q(t)\,e_Q(-t^2)].
\end{eqnarray}

Now we have all the necessary tools to prove the orthogonality relation for the $q$-Hermite polynomials. The proof can be found in \cite{MR0708496} or from the steps in the proof of theorem \ref{CHorth} using theorem \ref{hermrec}$(ii)$ and formula \eqref{hermexp}.
\begin{theorem}
When $q<1$, the $q$-hermite polynomials are orthogonal with respect to the inner product $\langle f|g\rangle = \int_{-\lambda_Q}^{\lambda_Q} d_qt\,f\overline{g}e_Q(-t^2)$ with ${\lambda_Q}^2=\frac{1}{1-Q}$,
\begin{eqnarray*}
\int_{-\lambda_Q}^{\lambda_Q} d_qt\,H^q_k(t)H^q_{l}(t)e_Q(-t^2)&=&\delta_{kl}2(q+1)^{k-1}q^{\frac{1}{2}(k+1)(k+2)}[k]_q!\Gamma_Q(\frac{1}{2}).
\end{eqnarray*}
\end{theorem}

\begin{remark}
The inner product defined above is only positive definite if one considers functions defined on the set of points $\{\pm \lambda_Q q^j|j\in\mN\}$.
\end{remark}

\subsection{Clifford-Hermite polynomials}
\label{sectionCHerm}
Inspired by the $q$-Hermite's equation in section \ref{sectionHerm} and the $q$-Dirac operator we define the $q$-deformed Clifford-Hermite polynomials as solutions of a $q$-Clifford-Hermite's equation. The Clifford-Hermite polynomials were introduced in \cite{MR926831}. We will not repeat their properties here, as they can be found from taking the limit $q\to 1$.

\begin{definition}
\label{defCHerm}
The $q$-Clifford-Hermite polynomials are of the form
\begin{equation}
\label{as}
H^q_{j,m,k}(\ux)M_k=\sum_{i=0}^{\lfloor j/2\rfloor}a_i^{j,k}\ux^{j-2i}M_k.
\end{equation}
with $M_k$ a spherical monogenic of degree $k$. They are eigenvectors of the $q$-Clifford Hermite's equation
\begin{eqnarray*}
[\Delta_q-(q+1)\ux\upx^q]f(\ux)&=&-(q+1)\lambda f(q\ux).
\end{eqnarray*}

The normalization is given by $a^{j,k}_{0}=(q+1)^j$.
\end{definition}

In this section we will use the notation $2\beta=m+2k$, assuming that we take $k$ fixed, and $Q=q^2$. By a quick calculation and theorem \ref{pxFischer} we find that the eigenvalues are given by ($j=2t$ or $j=2t+1$)
\begin{eqnarray*}
\lambda_{2t,m,k}&=&[2t]_qq^{-2t-k}=(q+1)[t]_QQ^{-t-k/2}\\
\lambda_{2t+1,m,k}&=&[2t+m+2k]_qq^{-2t-1-k}=(q+1)[t+\beta]_QQ^{-t-(k+1)/2}.
\end{eqnarray*}

The explicit form of the $q$-Clifford-Hermite polynomials is given by
\begin{lemma}
The coefficients of the Clifford-Hermite polynomials in definition \ref{defCHerm} are given by
\begin{eqnarray*}
a^{2t,k}_i&=&(q+1)^{2t}Q^{\frac{1}{2}i(i+1)}\binom{t}{i}_Q[t-1+\beta]_Q[t-2+\beta]_Q\cdots[t-i+\beta]_Q
\end{eqnarray*}

and
\begin{eqnarray*}
a_i^{2t+1,k}&=&(q+1)^{2t+1}Q^{\frac{1}{2}i(i+1)}\binom{t}{i}_Q[t+\beta]_Q[t-1+\beta]_Q\cdots[t-i+1+\beta]_Q.
\end{eqnarray*}
\end{lemma}

\begin{proof}
First we calculate, using equation \eqref{laplFischer}
\begin{eqnarray*}
(\px^q)^2a_{i-1}^{2t,k}\ux^{2t-2i+2}M_k&=&[2t-2i+2]_q[2t-2i+m+2k]_qa_{i-1}^{2t,k}\ux^{2t-2i}M_k
\end{eqnarray*}

and using lemma \ref{pxFischer}
\begin{eqnarray*}
\ux\upx^qa_{i}^{2t,k}\ux^{2t-2i}M_k&=&[2t-2i]_qa_{i}^{2t,k}\ux^{2t-2i}M_k.
\end{eqnarray*}

Substituting these results and $\lambda_{2t,m,k}=[2t]_qq^{-2t-k}$ in the differential equation leads to
\begin{eqnarray*}
[2t-2i+2]_q[2t-2i+m+2k]_qa_{i-1}^{2t,k}&=&(q+1)a_i^{2t,k}([2t]_qq^{-2t-k}q^{2t-2i+k}-[2t-2i]_q)\\
&=&(q+1)a_i^{2t,k}\frac{q^{2t-2i}-q^{-2i}-q^{2t-2i}+1}{q-1}\\
&=&(q+1)^2a_i^{2t,k}\frac{Q^{-i}(Q^i-1)}{Q-1},
\end{eqnarray*}

or
\begin{eqnarray*}
a_i^{2t,k}&=&Q^i\frac{[t-i+1]_Q[t-i+\beta]_Q}{[i]_Q}a_{i-1}^{2t,k}.
\end{eqnarray*}

Iterating this yields $a_i^{2t,k}$. The $a_i^{2t+1,k}$ are calculated in the same way.
\end{proof}

Using the $Q$-Gamma function leads to the explicit form of the $q$-Clifford-Hermite functions,
\begin{eqnarray*}
H^q_{2t,m,k}(\ux)M_k=(q+1)^{2t}\sum_{i=0}^{t}Q^{\frac{1}{2}i(i+1)}\binom{t}{i}_Q\frac{\Gamma_Q(t+m/2+k)}{\Gamma_Q(t-i+m/2+k)}\ux^{2t-2i}M_k
\end{eqnarray*}

and
\begin{eqnarray*}
H^q_{2t+1,m,k}(\ux)M_k=(q+1)^{2t+1}\sum_{i=0}^{t}Q^{\frac{1}{2}i(i+1)}\binom{t}{i}_Q\frac{\Gamma_Q(t+1+m/2+k)}{\Gamma_Q(t+1-i+m/2+k)}\ux^{2t-2i+1}M_k.
\end{eqnarray*}

We only defined the $\Gamma_Q$-function for $Q<1$, but for $Q>1$ the notation above can still be used to denote $\frac{[t+m/2+k]_Q!}{[t-i+m/2+k]_Q!}$. The $q$-Clifford-Hermite polynomials are connected with a $q$-deformation of the Laguerre polynomials in \cite{MR0708496}. We define $\cL_t^{\alpha}(\cdot|Q)$ by
\begin{eqnarray}
\label{CHL}
H^q_{2t,m,k}(\ux)&=&(q+1)^{2t}[t]_Q!\cL_t^{\frac{m}{2}+k-1}(r^2|Q)\\
H^q_{2t+1,m,k}(\ux)&=&(q+1)^{2t+1}[t]_Q!\ux\cL_t^{\frac{m}{2}+k}(r^2|Q).
\end{eqnarray}
These Laguerre polynomials are also related to those in \cite{MR0486697, MR0618759}, as we will show later. In particular we obtain a $q$-deformation of the classical relation between one dimensional Hermite and Laguerre polynomials,
\begin{eqnarray*}
H^q_{2t}(u)&=&(-1)^t(q+1)^{2t}[t]_{q^2}!\cL^{-\frac{1}{2}}_t(u^2|q^2).
\end{eqnarray*}

\begin{lemma}
\label{pxCH}
The $q$-Clifford-Hermite polynomials satisfy the following relation
\begin{eqnarray*}
\upx^qH^q_{j,m,k}M_k&=&C(j,m,k)H^q_{j-1,m,k}M_k
\end{eqnarray*}

with $C(2t,m,k)=(q+1)^2[t]_Q$ and $C(2t+1,m,k)=(q+1)^2[t+\beta]_Q$.
\end{lemma}
\begin{proof}
For the even case the lemma follows from considering the coefficients,
\begin{eqnarray*}
\upx^qa_i^{2t,k}\ux^{2t-2i}M_k&=&(q+1)[t-i]_Qa_i^{2t,k}\ux^{2t-1-2i}M_k\\
&=&(q+1)^2[t]_Qa_i^{2t-1,k}\ux^{2t-1-2i}M_k.
\end{eqnarray*}

The odd case is calculated similarly.
\end{proof}
The $q$-Clifford-Hermite polynomials can also be calculated using a recursion formula.
\begin{lemma}
\label{CHrecur1}
The $q$-Clifford-Hermite polynomials satisfy the recursion formula
\begin{eqnarray*}
H^q_{j+1,m,k}M_k&=&(q+1)\ux H^q_{j,m,k}M_k+D(j,m,k)H^q_{j-1,m,k}M_k
\end{eqnarray*}

with $D(2t,m,k)=(q+1)^2Q^{t+\beta}[t]_Q$ and $D(2t+1,m,k)=(q+1)^2Q^{t+1}[t+\beta]_Q$.
\end{lemma}
\begin{proof}
We prove this again by looking at the coefficients. They have to satisfy
\begin{eqnarray*}
a_i^{j+1,k}&=&(q+1)a_i^{j,k}+D(j,m,k)a_{i-1}^{j-1,k}.
\end{eqnarray*}

For $j=2t$ we obtain
\begin{eqnarray*}
1&=&(q+1)\frac{a^{2t,k}_i}{a^{2t+1,k}_{i}}+D(2t,m,k)\frac{a^{2t-1,k}_{i-1}}{a^{2t+1,k}_i}\\
&=&\frac{[t-i+\beta]_Q}{[t+\beta]_{Q}}+D(2t,m,k)\frac{[-i]_Q}{(-1)(q+1)^2[t]_Q[t+\beta]_Q}\\
&=&\frac{1}{[t+\beta]_Q}\left([t-i+\beta]_Q-Q^{t+\beta}[-i]_Q\right).
\end{eqnarray*}

The odd case is proven similarly.
\end{proof}
Similar to the one dimensional case there are two creation operators.

\begin{theorem}
\label{creaCH}
The $q$-Clifford-Hermite polynomials satisfy the following two relations
\begin{eqnarray*}
(i)\,H_{j,m,k}^qM_k&=&\left[q^{\sigma_j}\upx^q+(q+1)\ux\right]H^q_{j-1,m,k}M_k
\end{eqnarray*}

with $\sigma_{2t}=2t$ and $\sigma_{2t+1}=2t+2k+m$ and
\begin{eqnarray*}
(ii)\,H_{j,m,k}^q(q\ux)M_k&=&q^{j}\left[\upx^q+(q+1)\ux\right]H^q_{j-1,m,k}(\ux)M_k.
\end{eqnarray*}
\end{theorem}
\begin{proof}
These two equations can be found from combining lemma \ref{CHrecur1} with lemma \ref{pxCH} and from combining lemma \ref{pxCH} with definition \ref{defCHerm}.
\end{proof}

Using the definition of the $q$-exponential (\ref{defexp}) and the Leibniz rule in lemma \ref{pxr} yields
\begin{eqnarray}
\label{qexpdirac}
\upx^qe_Q(\ux^2)&=&e_Q(q^2\ux^2)[\upx^q+(q+1)\ux],
\end{eqnarray}
so theorem \ref{creaCH}$(ii)$ can be written as
\begin{eqnarray}
\label{qexpdiracch}
H^q_{j,m,k}(q\ux)\,M_k\,e_Q(q^2\ux^2)&=&q^{j}\,\upx^q\,H^q_{j-1,m,k}(\ux)\,M_k\,e_Q(\ux^2).
\end{eqnarray}

\begin{theorem}
\label{CHorth}
For $q<1$, with $\mR^m_q=\mB^m(\lambda_Q)$ and $\lambda_Q^2=\frac{1}{1-Q}$, the $q$-Clifford-Hermite polynomials are orthogonal with respect to the inner product 
\begin{eqnarray*}
\langle f|g\rangle =\int_{\mR_q^m}\,d_qV(\ux)\,\left[\overline{f}\,g\,e_Q(\ux^2)\right]_0.
\end{eqnarray*}

For the even Clifford-Hermite polynomials this means
\begin{eqnarray*}
\int_{\mR^m_q}d_qV(\ux)\left[\overline{H^q_{2j,m,k}M_k^{(p)}}H^q_{2t,m,k}M_l^{(r)}e_Q(\ux^2)\right]_0=\delta_{jt}\delta_{kl}\delta_{pr}(q+1)^{4j-1}Q^{(j+1)(j+\beta)}[j]_Q!\Gamma_Q(j+\beta),
\end{eqnarray*}

for the odd case
\begin{eqnarray*}
\int_{\mR^m_q}d_qV(\ux)\left[\overline{H^q_{2j+1,m,k}M_k^{(p)}}H^q_{2t+1,m,k}M_l^{(r)}e_Q(\ux^2)\right]_0=\delta_{jt}\delta_{kl}\delta_{pr}(q+1)^{4j+1}Q^{(j+1)(j+\beta+2)}[j]_Q!\Gamma_Q(j+\beta+1)
\end{eqnarray*}

and for the mixed case
\begin{eqnarray*}
\int_{\mR^m_q}d_qV(\ux)\left[\overline{H^q_{2j+1,m,k}M_k^{(p)}}H^q_{2t,m,k}M_l^{(r)}e_Q(\ux^2)\right]_0&=&0.
\end{eqnarray*}
\end{theorem}
\begin{proof}
Equation (\ref{HH}) implies that $k=l$ is necessary for the Clifford-Hermite polynomials not to be orthogonal. Using equation \eqref{intnorm2} and equation (\ref{qexpdiracch}) yields
\begin{eqnarray*}
\int_{\mR^m_q}d_qV(\ux)\,\overline{H^q_{j,m,k}M_k}\,H^q_{t,m,k}M_k\,e_Q(\ux^2)&=&q^{m+2k}\int_{q^{-1}\mR^m_q}d_qV(\ux)\,\overline{H^q_{j,m,k}(q\ux)M_k}\,H^q_{t,m,k}(q\ux)M_k\,e_Q(q^2\ux^2)\\
\end{eqnarray*}
\[=q^{m+2k+t}\int_{q^{-1}\mR^m_q}d_qV(\ux)\,\overline{H^q_{j,m,k}(q\ux)M_k}\,\left(\upx^qH^q_{t-1,m,k}(\ux)M_k\,e_Q(\ux^2)\right).
\]
Now we use theorem \ref{qCauchy} with $e_Q(-Q^{-1}\frac{1}{1-Q})=0$ (lemma \ref{nulpexp}) and lemma \ref{pxCH},
\begin{eqnarray*}
=q^{m+2k+t}\int_{q^{-1}\mR^m_q}d_qV(\ux)\,\overline{\left[q^\Gamma\upx^qH^q_{j,m,k}(q\ux)M_k\right]}\,H^q_{t-1,m,k}(q\ux)q^kM_k\,e_Q(q^2\ux^2)\\
=q^{m+3k+t+1}C(j,m,k)\int_{q^{-1}\mR^m_q}d_qV(\ux)\,\overline{\left[q^\Gamma H^q_{j-1,m,k}(q\ux)M_k\right]}\,H^q_{t-1,m,k}(q\ux)M_k\,e_Q(q^2\ux^2)\\
=q^{k+t+1}C(j,m,k)\int_{\mR^m_q}d_qV(\ux)\,\overline{\left[q^\Gamma H^q_{j-1,m,k}M_k\right]}\,H^q_{t-1,m,k}M_k\,e_Q(\ux^2).\\
\end{eqnarray*}
Substituting equation \eqref{1Gamma1} for the even case yields
\begin{eqnarray*}
\langle H^q_{2j,m,k}M_k|\,H^q_{2t,m,k}M_k\rangle&=&Q^{t+\frac{m}{2}+k}(q+1)^2[j]_Q\langle H^q_{2j-1,m,k}M_k|\,H^q_{2t-1,m,k}M_k\rangle
\end{eqnarray*}

and for the odd case
\begin{eqnarray*}
\langle H^q_{2j+1,m,k}M_k|\,H^q_{2t+1,m,k}M_k\rangle&=&Q^{t+1}(q+1)^2[j+k+\frac{m}{2}]_Q\langle H^q_{2j,m,k}M_k|\,H^q_{2t,m,k}M_k\rangle.
\end{eqnarray*}

The theorem follows from iterating these results and lemma \ref{gammagaus},
\begin{eqnarray*}
\int_{\mR^m_q}d_qV(\ux)\,\left[\overline{M_k^{(p)}}\,M_k^{(r)}\,e_Q(\ux^2)\right]_0&=&\int_{0}^{\infty}d_qr\,r^{m+2k-1}e_Q(-r^2)\delta_{pr}\\
&=&\frac{1}{q+1}Q^{\beta}\Gamma_Q(\beta)\delta_{pr}.
\end{eqnarray*}
\end{proof}

Finally we take a closer look at the even Clifford-Hermite polynomials $H_{2j,m,k}^qM_k$. Because $M_k\in\cH_k$, $M_k$ is of the form $\sum_{A}H_k^Aa_A$ with $a_A\in\mR_{0,m}$ and $H_k^A$ scalar spherical harmonics. From lemma \ref{pxCH} we find that $\Delta_q H^q_{2j,m,k}M_k=-C(2j,m,k)C(2j-1,m,k)H^q_{2j-2,m,k}M_k$. Because $\Delta_q$ and $H_{2j,m,k}$ are scalar this formula also holds for the each scalar part $H_{2j,m,k}^qH^A_k$. We define the scalar $q$-Clifford-Hermite polynomials as $H_{2j,m,k}^qH_k$ for $H_k$ a scalar spherical harmonic, the annihilation operator is given by
\begin{eqnarray}
\label{B}
\Delta_qH_{2j,m,k}^qH_k&=&-{(q+1)^4} [j]_Q[j+\frac{m}{2}+k-1]_QH_{2j-2,m,k}^qH_k.
\end{eqnarray}

In order to obtain the creation operator we apply theorem \ref{creaCH}$(ii)$,
\begin{eqnarray*}
H_{2j,m,k}M_k&=&q^{k-\mE}q^{2j}\left[\upx^q+(q+1)\ux\right]q^{k-\mE}q^{2j-1}\left[\upx^q+(q+1)\ux\right]H_{2j-2,m,k}M_k.
\end{eqnarray*}

Since this operator is again scalar, see lemma \ref{scalar}, this also holds for the scalar $q$-Clifford-Hermite polynomials,
\begin{eqnarray}
\label{A}
H_{2j,m,k}H_k&=&-Q^{2j+k-\mE-1}\left[\Delta_q-4E+q^2(q+1)^2r^2\right]H_{2j-2,m,k}H_k.
\end{eqnarray}

\subsection{Generalized Laguerre polynomials}

In the previous section we found $q$-deformed generalized Laguerre polynomials from the relation
\begin{eqnarray*}
H^q_{2t,m,k}(x)&=&(q+1)^{2t}[t]_Q!\cL_t^{\frac{m}{2}+k-1}(r^2|Q).
\end{eqnarray*}

For a general $\alpha>-1$ we define the $Q$-Laguerre polynomials as
\begin{eqnarray*}
\cL_t^{\alpha}(u|Q)&=&\sum_{i=0}^{t}Q^{\frac{1}{2}(t-i)(t-i+1)}\frac{(-u)^{i}}{[t-i]_Q![i]_Q!}\frac{\Gamma_Q(t+\alpha+1)}{\Gamma_Q(i+\alpha+1)}.
\end{eqnarray*}

These are the second type of Laguerre polynomials considered in \cite{MR0708496}. When we make the substitution $Q\to q^{-1}$, using $[k]_{q^{-1}}=q^{1-k}[k]_q$, we find
\begin{eqnarray*}
\cL_t^{\alpha}(u|q^{-1})&=&\sum_{i=0}^{t}q^{-\frac{1}{2}(t-i)(t-i+1)}\frac{(-u)^{i}}{[t-i]_{q^{-1}}![i]_{q^{-1}}!}\frac{[t+\alpha]_{q^{-1}}!}{[i+\alpha]_{q^{-1}}!}\\
&=&q^{-\frac{1}{2}t(t+1)-\alpha t}\sum_{i=0}^{t}q^{i(i+\alpha)}\frac{(-u)^{i}}{[t-i]_{q^{}}![i]_{q^{}}!}\frac{[t+\alpha]_{q^{}}!}{[i+\alpha]_{q^{}}!}.
\end{eqnarray*}

These are the $q$-Laguerre polynomials in \cite{MR0618759}, or with a different normalization in \cite{MR0486697}. In \cite{MR0708496} both the $q$-Laguerre polynomials, which are connected with the substitution ($q\leftrightarrow q^{-1}$) were studied. We could also have used a second type of $q$-Hermite polynomials (see \cite{MR0708496}) to generalize to the Clifford setting to obtain the $q$-Laguerre polynomials in \cite{MR0618759}. The $Q$-Laguerre polynomials can be defined as the solution of the $Q$-difference equation (see \cite{MR0708496})
\begin{eqnarray}
\label{defLag}
Q^{\alpha+1}u(\partial_u^Q)^2\cL_t^{\alpha}(u|Q)+([\alpha+1]_Q-u)\partial_u^Q\cL_t^{\alpha}(u|Q)&=&[-t]_Q\cL_t^{\alpha}(Qu|Q).
\end{eqnarray}
For $\alpha=\frac{m}{2}+k-1$ this is equivalent to the differential equation in definition \ref{defCHerm}. Equation (\ref{defLag}) can be written using the $q$-exponential
\begin{eqnarray}
\label{defLag2}
\partial_u^Q\left(e_Q(-u)u^{\alpha+1}\partial_u^Q\cL_t^{\alpha}(u|Q)\right)&=&[-t]_Qu^\alpha e_Q(-Qu)\cL_t^{\alpha}(Qu|Q).
\end{eqnarray}

Using this we can prove the orthogonality of the $Q$-Laguerre polynomials, which is another way to prove the orthogonality of the Clifford-Hermite polynomials.

\begin{theorem}
For $Q<1$ the $Q$-Laguerre polynomials for a fixed $\alpha>-1$ are orthogonal with respect to the inner product
\begin{eqnarray*}
\langle f|g\rangle_{\alpha}&=&\int_{0}^{\frac{1}{1-Q}}d_Qu \,u^{\alpha}f(u)g(u)e_Q(-u).
\end{eqnarray*}
\end{theorem}
\begin{proof}
Equation (\ref{defLag2}) leads to the following relation
\begin{equation*}
\partial_u^Q\left[\cL_j^{\alpha}(u|Q)e_Q(-u)u^{\alpha+1}\partial_u^Q\cL_t^{\alpha}(u|Q)-\cL_t^{\alpha}(u|Q)e_Q(-u)u^{\alpha+1}\partial_u^Q\cL_j^{\alpha}(u|Q)\right]
\end{equation*}
\begin{equation*}
\label{vrLagorth}
=([-t]_Q-[-j]_Q)u^\alpha\cL_j^{\alpha}(Qu|Q) e_Q(-Qu)\cL_t^{\alpha}(Qu|Q).
\end{equation*}
The orthogonality then follows from (\ref{partiele}). For $\alpha=\frac{m}{2}+k-1$ the result can also be found from theorem \ref{CHorth} and lemma \ref{intqinv}$(i)$.
\end{proof}

Finally we construct a family of realizations of $U_q\left(\mathfrak{su}(1|1)\right)$ for which the $q$-Laguerre polynomials will be the eigenvectors of their representations. We define 
\begin{eqnarray*}
A=\frac{Q^{-\mE-\frac{m}{2}}\left[\Delta_q-4E+q^2(q+1)^2r^2\right]}{(q+1)^2} &\mbox{and}& B=\frac{\Delta_q}{(q+1)^2}
\end{eqnarray*}

and write equations \eqref{A} and \eqref{B} in terms of the $q$-Laguerre polynomials,
\begin{eqnarray*}
A\cL_{j-1}^{\frac{m}{2}+k-1}(r^2|Q)H_k&=&-[j]_QQ^{1-2j-k-\frac{m}{2}}\cL_{j}^{\frac{m}{2}+k-1}(r^2|Q)H_k
\end{eqnarray*}

and
\begin{eqnarray*}
B\cL_{j}^{\frac{m}{2}+k-1}(r^2|Q)H_k&=&-[j+\frac{m}{2}+k-1]_Q\cL_{j-1}^{\frac{m}{2}+k-1}(r^2|Q)H_k.
\end{eqnarray*}

We define $C=[A,B]_Q$, from its definition we find
\begin{eqnarray*}
C\cL_{j}^{\frac{m}{2}+k-1}(r^2|Q)H_k&=&Q^{1-2j-k-\frac{m}{2}}\left([j+\frac{m}{2}+k-1]_Q[j]_Q-Q^{-1}[j+\frac{m}{2}+k]_Q[j+1]_Q\right)\cL_{j}^{\frac{m}{2}+k-1}(r^2|Q)H_k\\
&=&Q^{1-2j-k-\frac{m}{2}}\left(\frac{Q^{2j+\frac{m}{2}+k-1}+1-Q^{2j+\frac{m}{2}+k}-Q^{-1}}{(Q-1)^2}\right)\cL_{j}^{\frac{m}{2}+k-1}(r^2|Q)H_k\\
&=&-Q^{-2j-k-\frac{m}{2}}[2j+k+\frac{m}{2}]_Q\cL_{j}^{\frac{m}{2}+k-1}(r^2|Q)H_k.
\end{eqnarray*}

These calculations yield
\begin{eqnarray*}
(AC-Q^2CA)\cL_{j}^{\frac{m}{2}+k-1}(r^2|Q)H_k&=&(Q+1)A\cL_{j}^{\frac{m}{2}+k-1}(r^2|Q)H_k
\end{eqnarray*}

and
\begin{eqnarray*}
(CB-Q^2BC)\cL_{j}^{\frac{m}{2}+k-1}(r^2|Q)H_k&=&(Q+1)B\cL_{j}^{\frac{m}{2}+k-1}(r^2|Q)H_k.
\end{eqnarray*}

Since the scalar Clifford-Hermite polynomials constitute a basis for $\mR[x_1,\cdots,x_m]$ (lemma \ref{fischerdecomp}) this suffices to prove the following $\mathfrak{su}(1|1)_q$-relations,
\begin{eqnarray*}
\left[ A,B \right]_{Q} &=& C\\
\left[ A, C\right]_{Q^2} &=& (Q+1)A \\
\left[ C,B \right]_{Q^2} &=& (Q+1)B.
\end{eqnarray*}

Hence we obtain a family of representations of $\mathfrak{su}(1|1)_q$. This fits into the theory of relations between representations of quantum algebras and $q$-special functions, see e.g. \cite{MR1121848, MR1803885, MR2024824}. For every $k,m\in\mN$, we define the operator $A_k^m$ on the space of polynomials in one variable $\mR[t]$ by
\begin{eqnarray*}
\left[A_k^m f(t)\right](t=r^2) H_k(\ux)&=&Af(r^2)H_k(\ux) 
\end{eqnarray*}

with $\ux\in\mR^m$ and $H_k$ an arbitrary spherical harmonic of degree $k$. The operators $B^m_k$ and $C^m_k$ are defined similarly.

\begin{theorem}
For every $k,m\in\mN$, the set of operators $\{A_k^m,B_k^m,C_k^m\}$ generate the $\mathfrak{su}(1|1)_q$-quantum algebra. The basis $\{\cL_{j}^{\frac{m}{2}+k-1}(t|Q)|j\in\mN \}$ of $\mR[t]$ is the set of eigenvectors for this representation of $\mathfrak{su}(1|1)_q$. 
\end{theorem}

The $\mathfrak{su}(1|1)_q$ algebra appearing here can be linked with the version of $U_Q\left(\mathfrak{su}(1|1)\right)$ in \cite{MR2024824}. Define $J_0$ by the relation
\begin{eqnarray*}
C&=&-Q^{-2J_0}[2J_0]_Q,
\end{eqnarray*}
this implies $J_0\cL_{j}^{\frac{m}{2}+k-1}(r^2|Q)H_k=\frac{1}{2}\left(2j+\frac{m}{2}+k\right)\cL_{j}^{\frac{m}{2}+k-1}(r^2|Q)H_k$, so
\begin{eqnarray*}
\left[ J_0,A \right] =A &\mbox{and}&\left[ J_0, B\right]= -B.
\end{eqnarray*}

By defining $J_+=Q^{J_0}A$ and $J_-=qB$ we calculate
\begin{eqnarray*}
\left[ J_-,J_+ \right] &=& qBQ^{J_0}A-qQ^{J_0}AB\\
&=&qQ^{J_0}\left(QBA-AB\right)\\
&=&-qQ^{J_0}C\\
&=&\frac{Q^{J_0}-Q^{-J_0}}{Q^{1/2}-Q^{-1/2}}.
\end{eqnarray*}

This relation together with $\left[ J_0,J_{\pm} \right] = \pm J_{\pm}$, shows that $J_\pm$ and $J_0$ generate the $U_Q\left(\mathfrak{su}(1|1)\right)$ algebra in \cite{MR2024824}.

\section{Conclusion}
Our aim was to extend the existing $q$-calculus with a theory of partial derivatives in higher dimensions and a $q$-Laplace operator acting on functions in commuting variables. This was done by imposing four natural axioms that a $q$-Dirac operator should satisfy and led to a unique $q$-Dirac operator. Since this is vector operator, it implies the definition of $q$-partial derivatives. The $q$-Laplace operator was defined as minus the square of the $q$-Dirac operator and is scalar. 

The $q$-Dirac operator and the vector variable generate the quantum algebra $\mathfrak{osp}(1|2)_q$, the $q$-Laplace operator and the norm squared generate $\mathfrak{sl}_2(\mR)_q$, the even subalgebra of $\mathfrak{osp}(1|2)_q$. This $\mathfrak{sl}_2(\mR)_q$ already appeared in other $q$-calculus problems and in quantum Euclidean space. Since the $q$-Dirac and $q$-Laplace operator still possess their $Spin(m)$ and $SO(m)$ invariance, we obtained the Howe dual pairs $\left(Spin(m),\mathfrak{osp}(1|2)_q\right)$ and $\left(SO(m),\mathfrak{sl}_2(\mR)_q\right)$. This can be a starting point for the study of general Howe duality including quantum algebras and quantum groups.

The $q$-Laplace operator defines a $q$-Schr\"odinger equation. It is shown that the $SO(m)$-invariant $q$-Schr\"odinger equation on undeformed Euclidean space is equivalent with the $SO_q(m)$-invariant Schr\"odinger equation on quantum Euclidean space. This is an example of the interaction between quantum groups and $q$-calculus.

The $q$-difference equation for $q$-Hermite polynomials and the $q$-Dirac operator lead to a $q$-deformation of the Clifford-Hermite equation. The corresponding $q$-Clifford-Hermite polynomials have creation and annihilation operators and satisfy a recursion formula. These properties and a $q$-Cauchy formula lead to an orthogonality relation for the $q$-Clifford-Hermite polynomials. The $q$-Clifford-Hermite polynomials can be expressed in terms of the $q$-Laguerre polynomials. This identification leads to a realization of the $\mathfrak{su}_q(1|1)$ algebra action on $\mR[t]$. The weight vectors of this representation are $q$-Laguerre polynomials. This gives a new $q$-calculus interpretation to the appearance of quantum algebras in the representation theory of q-special functions.

\subsection*{Acknowledgment}
The authors would like to thank Hendrik De Bie for helpful suggestions and comments.

\end{document}